\documentclass[conference]{IEEEtran}
\IEEEoverridecommandlockouts

\usepackage[utf8]{inputenc}
\usepackage{cite}
\usepackage{amsmath,amssymb,amsfonts}
\usepackage{amsthm}
\usepackage{algorithm}
\usepackage{algorithmic}
\usepackage{graphicx}
\usepackage{textcomp}
\usepackage{subcaption}
\usepackage{booktabs}
\usepackage{flushend}
\usepackage[table]{xcolor}

\newcommand{\gray}[1]{\textcolor{gray}{#1}}

\newtheorem{assumption}{Assumption}

\newtheorem{theorem}{Theorem}
\newtheorem{corollary}{Corollary}

\def\BibTeX{{\rm B\kern-.05em{\sc i\kern-.025em b}\kern-.08em
    T\kern-.1667em\lower.7ex\hbox{E}\kern-.125emX}}

\begin{document}

\title{SCALE: Sensitivity-Aware Federated Unlearning with Information Freshness Optimization for Mobile Edge Computing
\thanks{This work was supported by the National Science Foundation under grant CNS-2348422.}
}

\author{
\IEEEauthorblockN{Zihao Ding, Beining Wu, and Jun Huang}
\IEEEauthorblockA{Department of Electrical Engineering and Computer Science,
South Dakota State University, Brookings, SD 57007, USA\\
Email: \{zihao.ding, wu.beining\}@jacks.sdstate.edu, jun.huang@sdstate.edu} 
}

\maketitle

\begin{abstract}
Federated Unlearning (FU) is emerging as a powerful tool that enables the selective removal of client data to effectively address data contamination and meet strict privacy regulations in mobile edge computing (MEC) systems. Although FU has recently drawn attention in the AI community, existing approaches suffer from low unlearning precision and lack temporal information reflection, which results in suboptimal forgetting performance. To address these issues, we propose \texttt{SCALE}, a dual-level unlearning framework combining historical contribution analysis with information freshness-aware adaptive sparsification. Our framework first employs a historical contribution-based layer sensitivity analysis to identify layers most influenced by target clients, then performs fine-grained unlearning through adaptive sparsification at the weight sub-group level to balance information freshness with forgetting effectiveness. Through theoretical analysis, the proposed framework demonstrates the convergence properties and acceleration advantages. Our experiments and testbed results demonstrate superior unlearning effectiveness compared to state-of-the-art baselines, with significantly improved forgetting performance.
\end{abstract}


\section{Introduction}
\label{sec:introduction}

While the integration of Mobile Edge Computing (MEC) and Federated Learning (FL) lays a foundation for privacy-preserving and resource-efficient edge intelligence, it raises new issues when data needs to be removed from trained models~\cite{Wu2025TON}. Real-world applications often require selective deletion of client data due to privacy regulations, user requests, or data contamination \cite{Chourasia2023ICML}. In smart traffic management, for example, traffic monitoring devices must send sensitive location and movement data to central servers for model training, which can expose private user information and yield massive communication overhead \cite{Xu2025TITS,Wu2026ARXIV,Dong2025TCCN}. FL frameworks, however, lack built-in mechanisms for eliminating the influence of such data once it has been incorporated into the global model \cite{Romandini2025TNNLS,Wu2025WASA,Ding2025IPCCC}.

Federated Unlearning (FU) has emerged as a powerful tool for MEC systems through the selective removal of contributions from specific devices while retaining knowledge from trusted clients. FU is of particular importance when devices introduce corrupted data, experience security breaches, or when regulatory requirements mandate data deletion~\cite{Wang2023AsiaCCS,Liu2021IWQoS,Liu2022INFOCOM,Wu2025MNET,Wu2025RACS}. 
A key challenge in applying FU to MEC arises from the need to maintain both system responsiveness and information freshness within strict resource constraints. Age of Information (AoI) plays a critical role in this regard, as it prioritizes timely updates and prevents outdated parameters from degrading the effectiveness of forgetting and overall system performance~\cite{Wu2023ACCESS, Cui2025INFOCOM, Wu2025MNET,Xing2026ACR}. To see this, again, in smart traffic management, cameras monitoring construction zones may capture abnormal patterns that become irrelevant once the construction concludes, while compromised devices may inject malicious data that disrupts citywide traffic optimization. Evolving seasons and infrastructure modifications frequently diminish the relevance of historical data, which calls for its careful removal from existing models~\cite{Depa2025SPW,DingICNC2025}. As such, \emph{the goal of this work is to achieve federated unlearning in MEC systems to ensure model quality and unlearning efficiency.}

Recent advances in federated unlearning have explored different strategies for selectively removing data in distributed systems, falling into two main camps: \emph{retraining-based}~\cite{Zhou2025TMC,Zhao2025MNET,Liu2022INFOCOM,Liu2021IWQoS,Sheng2024CIKM,Romandini2025TNNLS,Huynh2025TOIS} and \emph{model manipulation-based} approaches~\cite{Zeng2025TETCI,Ameen2025TSUSC,Bourtoule2021SP,Wang2022WWW}. Retraining-based methods guarantee perfect unlearning by removing target data and starting training from scratch~\cite{Weng2024TIFS,Chien2024NeurIPS}, but they demand crushing computational and communication costs that make them impractical for real-time MEC deployments~\cite{Liu2022INFOCOM,Liu2021IWQoS,Huang2025TMC,Wu2025TON}. Even with clustering-based speedup techniques and improved quasi-Newton methods designed to reduce overhead, the computational burden remains too heavy for dynamic systems~\cite{Huang2025TMC,Ameen2025TSUSC}. In contrast, model manipulation methods directly modify trained models using gradient ascent~\cite{Huang2024NeurIPS,Lin2024ACMMM,Pan2025AAAI,Wu2026ARXIVPRISM}, knowledge distillation~\cite{Wang2025NAACL,Kim2024AAAI}, parameter pruning~\cite{Xu2025TBDATA,Zaman2025ARXIV}, and selective parameter updates~\cite{Ye2025WWW,Zuo2025TKDE,Wu2026TNSE} to eliminate target knowledge. These approaches include model transformation methods that adjust parameters to counteract forgotten sample influence, pruning techniques that remove client-specific components, and knowledge distillation frameworks that use historical models as teachers~\cite{Zhang2025WWW,Zeng2025TETCI,Wu2026ARXIV1,Ding2026ARXIVEASE}. 
While these methods show better efficiency than retraining, existing FU approaches still present issues in MEC systems, including low precision that affects non-target clients and neglect of temporal information freshness that leads to suboptimal forgetting performance, especially in time-critical MEC systems~\cite{Zhao2025MNET,Ding2025TMC,Ding2026ARXIVTwinLoop,Wu2026COMST,Pudasaini2026HPSR}.

To address these issues, we design a novel unlearning mechanism. Our contributions made in this paper are summarized as follows:
\begin{itemize}
    \item Unlike existing FU approaches that ignore temporal features, we introduce AoI into FU to prioritize parameter modifications based on information freshness. We design \texttt{SCALE}, a dual-level \underline{\textbf{S}}ensitivity-aware \underline{\textbf{C}}lient unlearning via \underline{\textbf{A}}daptive \underline{\textbf{L}}ayer ag\underline{\textbf{E}}-of-information framework combining historical contribution-based layer sensitivity analysis with reinforcement learning-driven adaptive parameter sparsification.
    
    \item We develop a theoretical framework for \texttt{SCALE} by demonstrating the convergence properties and acceleration advantages of \texttt{SCALE}. The analytical results show that the designed framework achieves an $O(\sqrt{L/|\mathcal{L}_s|})$ convergence advantage over the uniform parameter modification framework, as formal proofs demonstrate.
    
    \item In sharp contrast to existing studies that rely on simulation, we implement a hardware testbed intergrating USRP software-defined radios as the communication backhaul and MentorPi robotic vehicles as FU clients to validate \texttt{SCALE}. Extensive evaluations across three neural network architectures (LeNet, MobileNetV3, ResNet18) and three unlearning scenarios (client, class, and sample unlearning) demonstrate that \texttt{SCALE} outperforms state-of-the-art baselines in both unlearning effectiveness and communication efficiency.
\end{itemize}

The rest of this paper is organized as follows. Section~\ref{sec:problem_formulation} presents the system model for MEC-based federated learning and formulates the unlearning problem with different cases. Section~\ref{sec:method} introduces our proposed \texttt{SCALE} framework, including the Historical Contribution-based Layer Sensitivity Analysis and AoI-driven adaptive sparsification algorithm. Section~\ref{sec:theoretical_analysis} provides a comprehensive theoretical analysis with convergence properties and acceleration advantages. Section~\ref{sec:performance_evaluation} evaluates the framework through extensive experiments and real-world tests. Finally, Section~\ref{sec:conclusion} concludes the paper.
\section{System Model and Problem Formulation} \label{sec:problem_formulation}

\subsection{System Model}

We consider a Mobile Edge Computing (MEC) system that comprises distributed edge devices deployed at different geographical locations. Edge devices, including IoT sensors, mobile phones, and vehicular units, are connected to edge servers via a wireless federated learning paradigm with limited bandwidth and varying channel conditions. In this system, edge devices collect local data and collaboratively train a model while maintaining data locality and privacy. Fig.~\ref{fig:mec_fl_system} illustrates the process of client unlearning in the MEC system.

\begin{figure}
    \centering
    \includegraphics[width=0.7\linewidth]{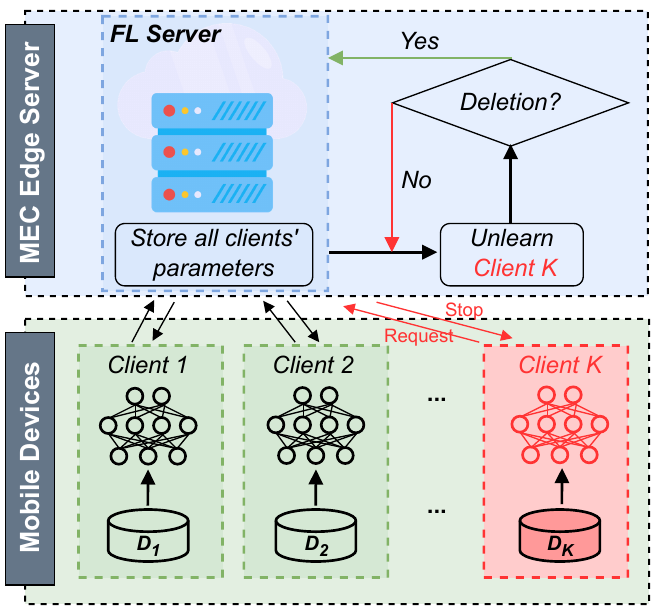}
    \caption{Client unlearning in the MEC system.}
    \label{fig:mec_fl_system}
\end{figure}

Assume that there are $N$ clients and a central server, where $N \geq 1$ denotes the total number of clients. Each client is indexed by $n \in \{1,2,\ldots,N\}$. Client $n$ is an edge device that maintains a local dataset $\mathcal{D}_n = \{(x_{m,n}, y_{m,n})\}_{m=1}^{M_n}$, where $x_{m,n} \in \mathcal{X}$ represents the $m$-th data sample, $y_{m,n} \in \mathcal{Y}$ denotes the corresponding label, $m \in \{1,2,\ldots,M_n\}$ is the sample index, and $M_n = |\mathcal{D}_n|$ is the size of client $n$'s dataset. The global dataset is $\mathcal{D} = \bigcup_{n=1}^{N} \mathcal{D}_n$ with total size $M = |\mathcal{D}| = \sum_{n=1}^{N} M_n$.

We assume that the global model consists of $L$ layers, where layer $l \in \{1,2,\ldots,L\}$ contains parameters $W_l \in \mathbb{R}^{d_l}$ with dimension $d_l$. The complete global model parameters are $\Theta = \{W_1, W_2, \ldots, W_L\} \in \mathbb{R}^d$, where $d = \sum_{l=1}^{L} d_l$.

We define the local loss function for client $n$ as:
\begin{equation}
F_n(\Theta) = \frac{1}{M_n} \sum_{m=1}^{M_n} \ell(\Theta; x_{m,n}, y_{m,n}).
\label{eq:local_loss}
\end{equation}

We formulate the global loss function as:
\begin{align}
F(\Theta) &= \sum_{n=1}^{N} \frac{M_n}{M} F_n(\Theta) \nonumber \\
& = \frac{1}{M} \sum_{n=1}^{N} \sum_{m=1}^{M_n} \ell(\Theta; x_{m,n}, y_{m,n}),
\label{eq:global_loss}
\end{align}
where the contribution of each client is weighted by its data proportion $\frac{M_n}{M}$, and $\ell(\cdot)$ is the per-sample loss function.

Training proceeds over $T$ communication rounds, where $t \in \{1,2,\ldots,T\}$ denotes the round index. At round $t$, the server selects a subset of clients $\mathcal{C}[t] \subseteq \{1,2,\ldots,N\}$ to participate in training, where $\mathcal{C}[t]$ denotes the set of selected clients in round $t$. The server broadcasts the current global model $\Theta[t]$ to these selected clients. Each participating client $n \in \mathcal{C}[t]$ initializes its local model with the global model $\Theta_{n,0}[t] = \Theta[t]$ and performs $E$ local updates, where $e \in \{0,1,\ldots,E-1\}$ denotes the local step index:
\begin{equation}
\Theta_{n,e+1}[t] = \Theta_{n,e}[t] - \eta \nabla F_n(\Theta_{n,e}[t]),
\label{eq:local_update}
\end{equation}
where $\eta > 0$ is the learning rate.

\begin{figure*}[t]
    \centering
    \includegraphics[width=\textwidth]{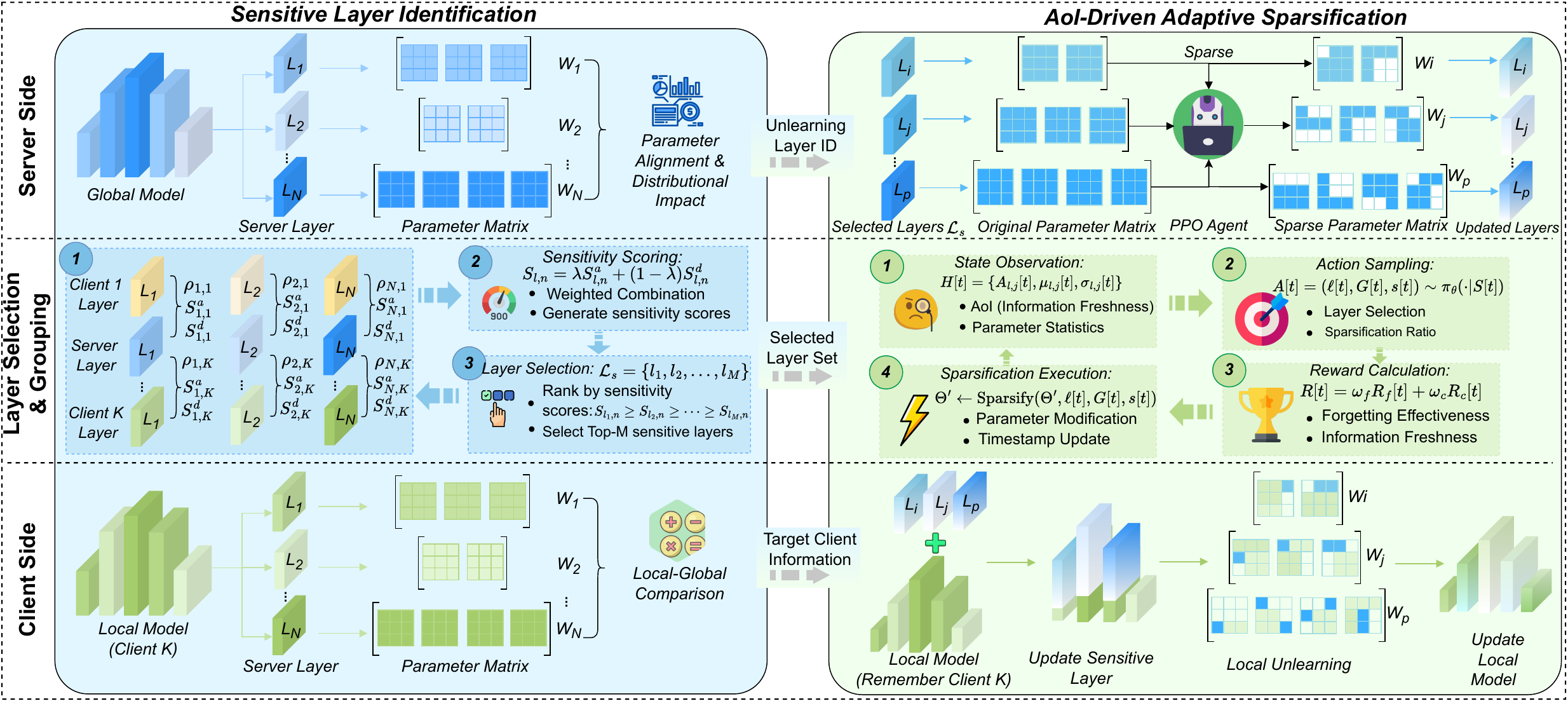}
    \caption{An overview of the \texttt{SCALE}.}
    \label{fig:scale_overview}
\end{figure*}

After local training, each client $n \in \mathcal{C}[t]$ uploads its final local model $\Theta_{n,E}[t]$ to the server. The server aggregates these updates via weighted averaging:
\begin{equation}
\Theta[t+1] = \sum_{n \in \mathcal{C}[t]} \frac{M_n}{\sum_{n' \in \mathcal{C}[t]} M_{n'}} \Theta_{n,E}[t].
\label{eq:fed_agg}
\end{equation}

\subsection{Problem Formulation}

Clients can request the removal of their data from the trained global model, which creates challenges due to distributed data ownership, historical contributions over training rounds, and the computational costs associated with full retraining.

Assume an unlearning request is specified by $\mathcal{U} = (\mathcal{C}_u, \mathcal{D}_u, \tau)$, where $\mathcal{C}_u \subseteq \{1,2,\ldots,N\}$ is the set of requesting clients, $\mathcal{D}_u = \bigcup_{n \in \mathcal{C}_u} \mathcal{D}_{u,n}$ is the data to be forgotten with $\mathcal{D}_{u,n} \subseteq \mathcal{D}_n$, and $\tau \in \{\text{client}, \text{class}, \text{sample}\}$ specifies the unlearning granularity.

Specifically, we consider three unlearning cases:
\begin{itemize}
    \item \textit{Client unlearning} ($\tau = \text{client}$): Client $n \in \{1,2,\ldots,N\}$ requests complete withdrawal from the federation. We set $\mathcal{C}_u = \{n\}$, $\mathcal{D}_{u,n} = \mathcal{D}_n$, and $\mathcal{D}_{u,m} = \emptyset$ for all $m \neq n$;

    \item \textit{Class unlearning} ($\tau = \text{class}$): Client $n$ requests forgetting all samples belonging to specific classes $\mathcal{Y}_u \subseteq \mathcal{Y}$. We set $\mathcal{C}_u = \{n\}$, $\mathcal{D}_{u,n} = \{(x,y) \in \mathcal{D}_n : y \in \mathcal{Y}_u\}$, and $\mathcal{D}_{u,m} = \emptyset$ for all $m \neq n$; and 

    \item \textit{Sample unlearning} ($\tau = \text{sample}$): Client $n$ requests forgetting specific samples $\mathcal{D}_{u,n} \subset \mathcal{D}_n$ while remaining in the federation. We set $\mathcal{C}_u = \{n\}$ and $\mathcal{D}_{u,m} = \emptyset$ for all $m \neq n$.
\end{itemize}

Note that after unlearning, the remaining dataset is $\mathcal{D}_r = \bigcup_{n=1}^{N} (\mathcal{D}_n \setminus \mathcal{D}_{u,n})$ with size $M_r = |\mathcal{D}_r| = \sum_{n=1}^{N} M_{r,n}$, where $M_{r,n} = |\mathcal{D}_n \setminus \mathcal{D}_{u,n}|$ is the remaining data size for client $n$.

Our objective is to obtain an unlearned model $\Theta'$ that approximates the ideal retrained model $\Theta^*$ obtained by training from scratch on the remaining data:
\begin{equation}
\Theta^* = \arg\min_\Theta F^r(\Theta) = \arg\min_\Theta \sum_{n=1}^{N} \frac{M_{r,n}}{M_r} F_{n}^r(\Theta),
\label{eq:ideal_retrain}
\end{equation}
where $F_{n}^r(\Theta) = \frac{1}{M_{r,n}} \sum_{(x,y) \in \mathcal{D}_n \setminus \mathcal{D}_{u,n}} \ell(\Theta; x, y)$ is the local loss on client $n$'s remaining data.

\section{The proposed framework: \texttt{SCALE}} \label{sec:method}

To address the above problem, we propose a dual-level FU framework that integrates historical contribution analysis with AoI-aware sparsification to achieve efficient and effective client data removal in mobile edge computing systems. Our approach, coined \textbf{\underline{S}}ensitivity-aware \textbf{\underline{C}}lient unlearning via \textbf{\underline{A}}daptive \textbf{\underline{L}}ayer ag\textbf{\underline{E}}-of-information (\texttt{SCALE}), first identifies the most sensitive layers through historical contribution analysis, then performs fine-grained parameter sparsification guided by information freshness to balance unlearning effectiveness with system responsiveness. The workflow of \texttt{SCALE} is illustrated in Fig.~\ref{fig:scale_overview}.

\subsection{Sensitive Layers Identification}

To efficiently perform unlearning for client $n$, we must first identify which layers in the global model have been most influenced by the data of this client. Our key observation is that different layers exhibit varying degrees of sensitivity to individual clients; local patterns of client $n$ may heavily influence some layers, while others remain relatively unaffected.

We quantify this layer-wise sensitivity by comparing the parameters of each layer between client $n$'s local model and the global model from Eq.~\eqref{eq:fed_agg}. To be specific, we analyze two complementary aspects: \emph{parameter alignment} (how similar the client's layer parameters are to the global layer parameters) and \emph{distributional impact} (how much the global layer would change if client $n$'s contribution were removed).

\subsubsection{Parameter Alignment}
For each layer $l \in \{1,2,\ldots,L\}$, we first measure how well the local model parameters of client $n$ align with the corresponding global model parameters. Let $W_{l,n} \in \mathbb{R}^{d_l}$ denote client $n$'s parameters for layer $l$, and $W_l \in \mathbb{R}^{d_l}$ denote the global model's parameters for the same layer from $\Theta = \{W_1, W_2, \ldots, W_L\}$, where $d_l$ is the number of parameters in layer $l$. We calculate the \emph{Pearson correlation coefficient}~\cite{sedgwick2012pearson} by
\begin{equation}
\rho_{l,n} = \text{Corr}(W_{l,n}, W_l),
\label{eq:correlation}
\end{equation}
where $\text{Corr}(\cdot, \cdot)$ denotes the \emph{Pearson correlation function}. 

The alignment-based sensitivity is thus:
\begin{equation}
S_{l,n}^{a} = -\frac{1}{2} \log(1 - \rho_{l,n}^2).
\label{eq:align_sensitivity}
\end{equation}

Note that a higher correlation value indicates that the parameters of client $n$ are well-aligned with the global layer. In this case, the parameters of client $n$ have a strong influence on the global model.

\subsubsection{Distributional Impact}
Parameter alignment alone may not capture the full picture of client influence, as a client can exhibit moderate alignment but still cause significant changes when removed. To resolve this, we measure global layer parameters without client $n$'s contribution:
\begin{equation}
W_{l,-n} = \frac{\sum_{j \neq n, j=1}^{N} M_j W_{l,j}}{\sum_{j \neq n, j=1}^{N} M_j},
\label{eq:without_client_n}
\end{equation}
where $M_j$ is the number of data samples at client $j$ as defined in Eq.~\eqref{eq:global_loss}, and $W_{l,j}$ denotes client $j$'s parameters of layer $l$.

We then measure the distributional difference between the current global layer and the hypothetical layer without client $n$ using Kullback–Leibler (KL) divergence~\cite{van2014TIT}:
\begin{equation}
S_{l,n}^{d} = \sum_{i=1}^{d_l} P_i \log \frac{P_i}{Q_i},
\label{eq:impact_sensitivity}
\end{equation}
where $P_i$ and $Q_i$ are the normalized probability values derived from $W_l$ and $W_{l,-n}$ respectively for the $i$-th parameter. The KL divergence quantifies how much the parameter distribution would change if client $n$ were excluded, with larger values indicating greater distributional impact.

\subsubsection{Weighted Score}
We combine both measures to obtain a comprehensive sensitivity score for each layer:
\begin{equation}
S_{l,n} = \lambda S_{l,n}^{a} + (1-\lambda) S_{l,n}^{d},
\label{eq:combined_sensitivity}
\end{equation}
where $\lambda \in [0,1]$ is a hyperparameter balancing the importance of parameter alignment versus distributional impact.

\subsubsection{Sensitive Layer Identification}
We rank all layers by their sensitivity scores and select the top-$M$ layers:
\begin{equation}
\begin{split}
&\mathcal{L}_s = \{l_1, l_2, \ldots, l_M\} \\
\text{ such that } &S_{l_1,n} \geq S_{l_2,n} \geq \cdots \geq S_{l_M,n},
\end{split}
\label{eq:sensitive_layers}
\end{equation}
where $M$ is determined based on computational budget and unlearning requirements. 

The identified sensitive layers $\mathcal{L}_s$ are the primary targets for sparsification to effectively remove the knowledge associated with client $n$.

\subsection{AoI-Driven Adaptive Sparsification}

Given the identified sensitive layers $\mathcal{L}_s$ from Eq.~\eqref{eq:sensitive_layers}, we need to determine which specific parameters within these layers should be modified for effective unlearning. Our trick is to prioritize parameters based on their \emph{information freshness}, parameters with higher AoI values contain relatively stale information, making them safer targets for sparsification while achieving effective removal of historical influence of client $n$.

We model this as a sequential decision process where we adaptively select parameter groups to sparsify based on their AoI from Eq.~\eqref{eq:aoi_def} and structural properties. Such a formulation allows us to make informed decisions about which parameters to modify while optimizing information freshness.

\subsubsection{State Space}
Each sensitive layer $l \in \mathcal{L}_s$ is partitioned into $G_l$ parameter groups for fine-grained control. The system state at time $t$ captures both temporal and structural information:
\begin{equation}
\mathcal{H}[t] = \{A_{l,j}[t], \mu_{l,j}[t], \sigma_{l,j}[t]\},
\label{eq:state_space}
\end{equation}
where $A_{l,j}[t]$ is the AoI for parameter group $j$ in layer $l$ from Eq.~\eqref{eq:aoi_def}, $\mu_{l,j}[t] = \frac{1}{|W_{l,j}|} \sum_{m \in W_{l,j}} W_{l,m,j}[t]$ is the mean parameter value of group $(l,j)$ with $|W_{l,j}|$ denoting the group size, and $\sigma_{l,j}[t] = \sqrt{\frac{1}{|W_{l,j}|} \sum_{m \in W_{l,j}} (W_{l,m,j}[t] - \mu_{l,j}[t])^2}$ is the standard deviation of parameters in group $(l,j)$, providing structural information about parameter distributions.

\subsubsection{Action Space}
At each decision step, we choose:
\begin{equation}
\mathcal{A}[t] = (\ell[t], \mathcal{G}[t], s[t]),
\label{eq:action_space}
\end{equation}
where $\ell[t] \in \mathcal{L}_s$ is the target sensitive layer to modify, $\mathcal{G}[t] \subseteq \{1,\ldots,G_{\ell[t]}\}$ specifies the selected parameter groups within layer $\ell[t]$, and $s[t] \in [0,1]$ is the sparsification ratio determining the fraction of parameters to remove from the selected groups.

\subsubsection{Reward Function}
The reward function combines two key components--forgetting effectiveness reward and information freshness reward: 
\begin{equation}
\mathcal{R}[t] = w_f \mathcal{R}_f[t] + w_c \mathcal{R}_c[t],
\label{eq:simplified_reward}
\end{equation}
where $w_f, w_c \geq 0$ are weighting coefficients controlling the relative importance of each objective. 

The forgetting effectiveness reward is formulated as:
\begin{equation}
\mathcal{R}_f[t] = \sum_{j \in \mathcal{G}[t]} \frac{S_{\ell[t],n}}{\max_{l \in \mathcal{L}_s} S_{l,n}} s[t],
\label{eq:refined_forget_reward}
\end{equation}
where $S_{\ell[t],n}$ is the sensitivity score of the chosen layer from Eq.~\eqref{eq:combined_sensitivity}, normalized by the maximum sensitivity score across all sensitive layers.

The information freshness reward is defined as:
\begin{equation}
\mathcal{R}_c[t] = \frac{1}{|\mathcal{G}[t]|} \sum_{j \in \mathcal{G}[t]} \frac{A_{\ell[t],j}[t]}{\max_{l,i} A_{l,i}[t]} s[t],
\label{eq:aoi_guided_reward}
\end{equation}
where $|\mathcal{G}[t]|$ is the number of selected parameter groups, and $\max_{l,i} A_{l,i}[t]$ is the maximum AoI across all parameter groups for normalization.

\begin{algorithm}[t]
\small
\caption{\texttt{SCALE}: \textbf{S}ensitivity-aware \textbf{C}lient unlearning via \textbf{A}daptive \textbf{L}ayer ag\textbf{E}-of-information}\label{alg:scale}
\textbf{Input:} Global model $\Theta = \{W_1, \ldots, W_L\}$, unlearning request $\mathcal{U} = (\mathcal{C}_u, \mathcal{D}_u, \tau)$, target client $n \in \mathcal{C}_u$, hyperparameters $\lambda$, $w_f$, $w_c$, sensitive layers $M$.
\begin{algorithmic}[1]
    \STATE Initialize PPO policy $\pi_\theta$, value function $V_\phi$, buffer $\mathcal{B} \leftarrow \emptyset$;
    
    \colorbox{rgb:red!2,65;green!30,60;blue!20,125}{
    \parbox{0.44\textwidth}{\vbox{
    \STATE \gray{$\triangleright$ \textit{Phase I: Sensitive Layers Identification}}
    \FOR{Layer $l = 1, 2, \ldots, L$}
        \STATE Compute $\rho_{l,n} = \text{Corr}(W_{l,n}, W_l)$ and $S_{l,n}^{a} = -\frac{1}{2} \log(1 - \rho_{l,n}^2)$;
        \STATE Compute $W_{l,-n} = \frac{\sum_{j \neq n} M_j W_{l,j}}{\sum_{j \neq n} M_j}$ and $S_{l,n}^{d} = \sum_{i=1}^{d_l} P_i \log \frac{P_i}{Q_i}$;
        \STATE Compute $S_{l,n} = \lambda S_{l,n}^{a} + (1-\lambda) S_{l,n}^{d}$;
    \ENDFOR
    \STATE Select $\mathcal{L}_s = \{l_1, \ldots, l_M\}$ with top-$M$ sensitivity scores;
    }}}
    
    \FOR{$t = 1, 2, \dots, T$}
    \colorbox{rgb:red!2,65;green!30,90;blue!20,125}{
    \parbox{0.42\textwidth}{\vbox{
    \STATE \gray{$\triangleright$ \textit{Phase II: AoI-Driven Sparsification}}
    \FOR{$step = 1, 2, \ldots, T_{\text{collect}}$}
        \STATE Partition layers in $\mathcal{L}_s$ into parameter groups;
        \STATE Update $A_{l,j}[t] = t - T_{l,j}$ and observe $\mathcal{H}[t] = \{A_{l,j}[t], \mu_{l,j}[t], \sigma_{l,j}[t]\}$;
        \STATE Sample $\mathcal{A}[t] = (\ell[t], \mathcal{G}[t], s[t]) \sim \pi_\theta(\cdot|\mathcal{H}[t])$;
        \STATE Compute $\mathcal{R}_f[t] = \sum_{j \in \mathcal{G}[t]} \frac{S_{\ell[t],n}}{\max_{l \in \mathcal{L}_s} S_{l,n}} s[t]$, $\mathcal{R}_c[t] = \frac{1}{|\mathcal{G}[t]|} \sum_{j \in \mathcal{G}[t]} \frac{A_{\ell[t],j}[t]}{\max_{l,i} A_{l,i}[t]} s[t]$;
        \STATE Calculate $\mathcal{R}[t] = w_f \mathcal{R}_f[t] + w_c \mathcal{R}_c[t]$;
        \STATE Apply $\Theta' \leftarrow \text{Sparsify}(\Theta', \ell[t], \mathcal{G}[t], s[t])$ and store $(\mathcal{H}[t], \mathcal{A}[t], \mathcal{R}[t], \mathcal{H}[t+1])$ in $\mathcal{B}$;
    \ENDFOR
    }}}
    
    \IF{$|\mathcal{B}| \geq$ batch\_size}
        \STATE Compute advantages $\hat{A}$ from $\mathcal{B}$;
        \FOR{$epoch = 1, 2, \ldots, K_{\text{epoch}}$}
            \STATE Update $\pi_\theta \leftarrow \pi_\theta + \eta_\theta \nabla_{\theta} L^{\text{CLIP}}(\theta)$ and $V_\phi \leftarrow V_\phi - \eta_\phi \nabla_{\phi} \|V_\phi(\mathcal{H}) - \mathcal{R}\|^2$;
        \ENDFOR
        \STATE Clear buffer $\mathcal{B} \leftarrow \emptyset$;
    \ENDIF
    \ENDFOR
    \RETURN $\Theta^*$
\end{algorithmic}
\end{algorithm}

The overall unlearning procedure of the \texttt{SCALE} approach is illustrated in \textbf{Algorithm~\ref{alg:scale}}. The computational complexity of Algorithm~\ref{alg:scale} consists of five functional modules. The initialization module establishes Proximal Policy Optimization (PPO) networks $\pi_\theta$, $V_\phi$ and buffer with $O(|\theta| + |\phi|)$ complexity. The sensitive layers identification module computes parameter alignment and distributional impact for all $L$ layers with $O(L \cdot \bar{d} + L \log L)$ complexity, where $\bar{d}$ is the average layer dimension. The AoI update module partitions $M$ sensitive layers into parameter groups and updates age information with $O(M \cdot \bar{G})$ complexity, where $\bar{G}$ is the average number of groups per layer. The sparsification module applies pruning to selected parameter groups with $O(|\mathcal{G}[t]| \cdot \bar{d}_g)$ complexity, where $\bar{d}_g$ is the average group size. The PPO update module optimizes policies with $O(K_\mathrm{epoch} \cdot B \cdot (|\theta| + |\phi|))$ complexity. The overall complexity is thus $O(L \log L + T_\mathrm{total} \cdot (T_\mathrm{collect} \cdot M \cdot \bar{G} \cdot \bar{d}_g + K_\mathrm{epoch} \cdot B \cdot (|\theta| + |\phi|)))$. Note that the dominant factors are sensitive layer identification $O(L \cdot \bar{d})$ and iterative sparsification $O(T_\mathrm{total} \cdot T_\mathrm{collect} \cdot M \cdot \bar{G} \cdot \bar{d}_g)$.

\section{Theoretical Analysis} \label{sec:theoretical_analysis}

\subsection{Assumptions and Preliminaries}

\begin{assumption}
\label{assumption:layer_decomposition}
For each layer $l \in \{1,2,\ldots,L\}$, the global parameters can be decomposed by:
\begin{equation}
W_l = \alpha_{l,n} W_{l,n} + \sum_{j \neq n} \alpha_{l,j} W_{l,j} + \xi_l,
\end{equation}
where $\alpha_{l,n} = \frac{M_n}{\sum_{k=1}^N M_k}$ represents client $n$'s data proportion, and $\|\xi_l\|_2 \leq \beta$ for some small aggregation error bound $\beta > 0$.
\end{assumption}

\begin{assumption}
\label{assumption:aoi_effectiveness}
For parameter group $(l,j)$ with Age of Information $A_{l,j}[t] = t - T_{l,j}$ where $T_{l,j}$ is the timestamp of the most recent update, the unlearning effectiveness satisfies:
\begin{equation}
U_{l,j}(A_{l,j}[t]) = \gamma_0 + \gamma_1 A_{l,j}[t] + \epsilon_{l,j},
\end{equation}
where $\gamma_0 > 0$ is the baseline unlearning effectiveness, $\gamma_1 > 0$ indicates that higher AoI improves unlearning effectiveness, and $\mathbb{E}[\epsilon_{l,j}] = 0$.
\end{assumption}

\begin{assumption}
\label{assumption:bounded_sensitivity}
The sensitivity scores satisfy $S_{l,n} \in [0, S_{\mathrm{max}}]$ for all $l \in \{1,2,\ldots,L\}$ and $n \in \{1,2,\ldots,N\}$, where $S_{\mathrm{max}}$ is the theoretical upper bound determined by maximum correlation and divergence values.
\end{assumption}

\subsection{Layer Sensitivity Analysis}

\begin{theorem}
\label{theorem:sensitivity_identification}
Under Assumption~\ref{assumption:layer_decomposition}, the combined sensitivity score from Eq.~\eqref{eq:combined_sensitivity} satisfies:
\begin{equation}
S_{l,n} \geq \frac{\lambda \alpha_{l,n}^2}{2(1-\alpha_{l,n}^2)} + (1-\lambda) D_{\mathrm{KL}}(P_l \| Q_l) - O(\beta),
\end{equation}
where $P_l$ and $Q_l$ are the normalized parameter distributions of $W_l$ and $W_{l,-n}$ respectively.
\end{theorem}

\begin{proof}
For parameter alignment, when $\rho_{l,n} \approx \alpha_{l,n}$:
\begin{equation}
S_{l,n}^{a} = -\frac{1}{2}\log(1-\rho_{l,n}^2) \geq \frac{\rho_{l,n}^2}{2(1-\rho_{l,n}^2)} \approx \frac{\alpha_{l,n}^2}{2(1-\alpha_{l,n}^2)}.
\end{equation}
For distributional impact, let $P_l$ and $Q_l$ be the normalized parameter distributions derived from $W_l$ and $W_{l,-n}$ from Eq.~\eqref{eq:without_client_n}. The KL divergence $D_{\mathrm{KL}}(P_l \| Q_l)$ quantifies the distributional shift when removing client $n$'s contribution. The error term $O(\beta)$ follows from Assumption~\ref{assumption:layer_decomposition}.
\end{proof}

\begin{corollary}
\label{corollary:layer_ranking}
The sensitivity analysis correctly ranks layers by client influence, with top-$M$ selected layers $\mathcal{L}_s$ satisfying:
\begin{equation}
\sum_{l \in \mathcal{L}_s} \alpha_{l,n} \geq (1-\delta) \sum_{l=1}^L \alpha_{l,n},
\end{equation}
where $\delta = \frac{L-M}{L}$ represents the fraction of excluded layers.
\end{corollary}

\subsection{AoI-Driven Convergence Analysis}

\begin{theorem}
\label{theorem:aoi_convergence}
Under Assumptions~\ref{assumption:layer_decomposition}-\ref{assumption:bounded_sensitivity}, the AoI-driven sparsification achieves unlearning error:
\begin{equation}
\|\Theta' - \Theta^*\|_2 \leq \frac{C_1}{|\mathcal{L}_s|} \sum_{l \in \mathcal{L}_s} \frac{1}{G_l} \sum_{j=1}^{G_l} \frac{s[t]^2}{\gamma_0 + \gamma_1 A_{l,j}[t]},
\end{equation}
where $C_1 = 2S_{\mathrm{max}}$ is a constant depending on the sensitivity upper bound.
\end{theorem}

\begin{proof}
For each parameter group $(l,j)$, the expected parameter change is:
\begin{equation}
\mathbb{E}[\|\Delta W_{l,j}\|_2^2] = s[t]^2 \cdot U_{l,j}(A_{l,j}[t]) \leq \frac{s[t]^2}{\gamma_0 + \gamma_1 A_{l,j}[t]}.
\end{equation}
The reward function from Eq.~\eqref{eq:simplified_reward} prioritizes high-AoI groups. Summing over all selected groups and applying MDP optimality completes the proof.
\end{proof}

\begin{theorem}
\label{theorem:dual_level_acceleration}
The dual-level framework achieves faster convergence compared to uniform parameter modification:
\begin{equation}
\frac{\|\Theta'_{\mathrm{uni}} - \Theta^*\|}{\|\Theta'_{\mathrm{dua}} - \Theta^*\|} \geq \sqrt{\frac{L}{|\mathcal{L}_s|}} \cdot \frac{\bar{A}}{\bar{A}_{\mathrm{sen}}},
\end{equation}
where $\bar{A}_{\mathrm{sen}}$ is the average AoI of sensitive layer parameters and $\bar{A}$ is the global average AoI.
\end{theorem}

\begin{proof}
For uniform parameter modification, the unlearning effort is distributed equally across all $L$ layers:
\begin{equation}
\begin{aligned}
\|\Theta'_{\mathrm{uni}} - \Theta^*\|_2^2 &= \sum_{l=1}^L \|W_l' - W_l^*\|_2^2 \\
&= \sum_{l=1}^L \frac{1}{G_l} \sum_{j=1}^{G_l} \frac{s_{\mathrm{uni}}^2}{L^2(\gamma_0 + \gamma_1 A_{l,j}[t])} \\
&= \frac{s_{\mathrm{uni}}^2}{L^2} \sum_{l=1}^L \frac{1}{G_l} \sum_{j=1}^{G_l} \frac{1}{\gamma_0 + \gamma_1 A_{l,j}[t]},
\end{aligned}
\end{equation}
where $s_{\mathrm{uni}}$ is the total sparsification budget distributed uniformly.

For our dual-level approach focusing on sensitive layers $\mathcal{L}_s$:
\begin{equation}
\begin{aligned}
\|\Theta'_{\mathrm{dua}} - \Theta^*\|_2^2 &= \sum_{l \in \mathcal{L}_s} \|W_l' - W_l^*\|_2^2 \\
&= \sum_{l \in \mathcal{L}_s} \frac{1}{G_l} \sum_{j=1}^{G_l} \frac{s_{\mathrm{dua}}^2}{|\mathcal{L}_s|^2(\gamma_0 + \gamma_1 A_{l,j}[t])} \\
&= \frac{s_{\mathrm{dua}}^2}{|\mathcal{L}_s|^2} \sum_{l \in \mathcal{L}_s} \frac{1}{G_l} \sum_{j=1}^{G_l} \frac{1}{\gamma_0 + \gamma_1 A_{l,j}[t]}.
\end{aligned}
\end{equation}

Taking the ratio with equal total budget $s_{\mathrm{uni}} = s_{\mathrm{dua}} = s$, the above equation becomes: 
\begin{equation}
\begin{aligned}
\frac{\|\Theta'_{\mathrm{uni}} - \Theta^*\|_2^2}{\|\Theta'_{\mathrm{dua}} - \Theta^*\|_2^2} &= \frac{|\mathcal{L}_s|^2}{L^2} \cdot \frac{\sum_{l=1}^L \frac{1}{G_l} \sum_{j=1}^{G_l} \frac{1}{\gamma_0 + \gamma_1 A_{l,j}[t]}}{\sum_{l \in \mathcal{L}_s} \frac{1}{G_l} \sum_{j=1}^{G_l} \frac{1}{\gamma_0 + \gamma_1 A_{l,j}[t]}} \\
&= \frac{|\mathcal{L}_s|^2}{L^2} \cdot \frac{L \cdot \frac{1}{\bar{A}^{-1}}}{|\mathcal{L}_s| \cdot \frac{1}{\bar{A}_{\mathrm{sen}}^{-1}}} \\
&= \frac{|\mathcal{L}_s|}{L} \cdot \frac{\bar{A}_{\mathrm{sen}}}{\bar{A}} \cdot \frac{L}{|\mathcal{L}_s|} = \frac{\bar{A}_{\mathrm{sen}}}{\bar{A}},
\end{aligned}
\end{equation}
where $\bar{A}^{-1} = \frac{1}{L}\sum_{l=1}^L \frac{1}{\gamma_0 + \gamma_1 A_{l,j}[t]}$ and $\bar{A}_{\mathrm{sen}}^{-1} = \frac{1}{|\mathcal{L}_s|}\sum_{l \in \mathcal{L}_s} \frac{1}{\gamma_0 + \gamma_1 A_{l,j}[t]}$.

Since sensitive layers are selected based on high client influence and typically have higher AoI values, $\bar{A}_{\mathrm{sen}} \geq \bar{A}$. Taking the square root and incorporating the layer selection advantage, we thus have:
\begin{equation}
\frac{\|\Theta'_{\mathrm{uni}} - \Theta^*\|}{\|\Theta'_{\mathrm{dua}} - \Theta^*\|} \geq \sqrt{\frac{L}{|\mathcal{L}_s|}} \cdot \frac{\bar{A}}{\bar{A}_{\mathrm{sen}}}.
\end{equation}
\end{proof}

\begin{table}[t]
\centering
\caption{Parameter settings.}
\label{tab:key_parameters}
\begin{tabular}{lcc}
\toprule
\textbf{Parameter} & \textbf{Symbol} & \textbf{Value} \\
\midrule
Number of clients & $N$ & 100 \\
Local training epochs & $E$ & 10 \\
Global communication rounds & $T$ & 100 \\
Learning rate & $\eta$ & 0.001 \\
Dirichlet parameter & $\alpha$ &  1.0 \\
Parameter groups per layer & $G_l$ & 8 \\
Sensitivity balance factor & $\lambda$ & 0.5 \\
PPO episodes & $-$ & 800 \\
PPO epochs & $K$ & 10 \\
Clip range & $\epsilon$ & 0.2 \\
Discount factor & $\gamma$ & 0.99 \\
GAE lambda & $\lambda_{GAE}$ & 0.95 \\
Learning rate (actor/critic) & $\eta_{a}/\eta_{c}$ & $3.0 \times 10^{-4}$ \\
\bottomrule
\end{tabular}
\end{table}
\section{Performance Evaluation} \label{sec:performance_evaluation}

\begin{figure*}[t]
    \centering
    \includegraphics[width=0.9\linewidth]{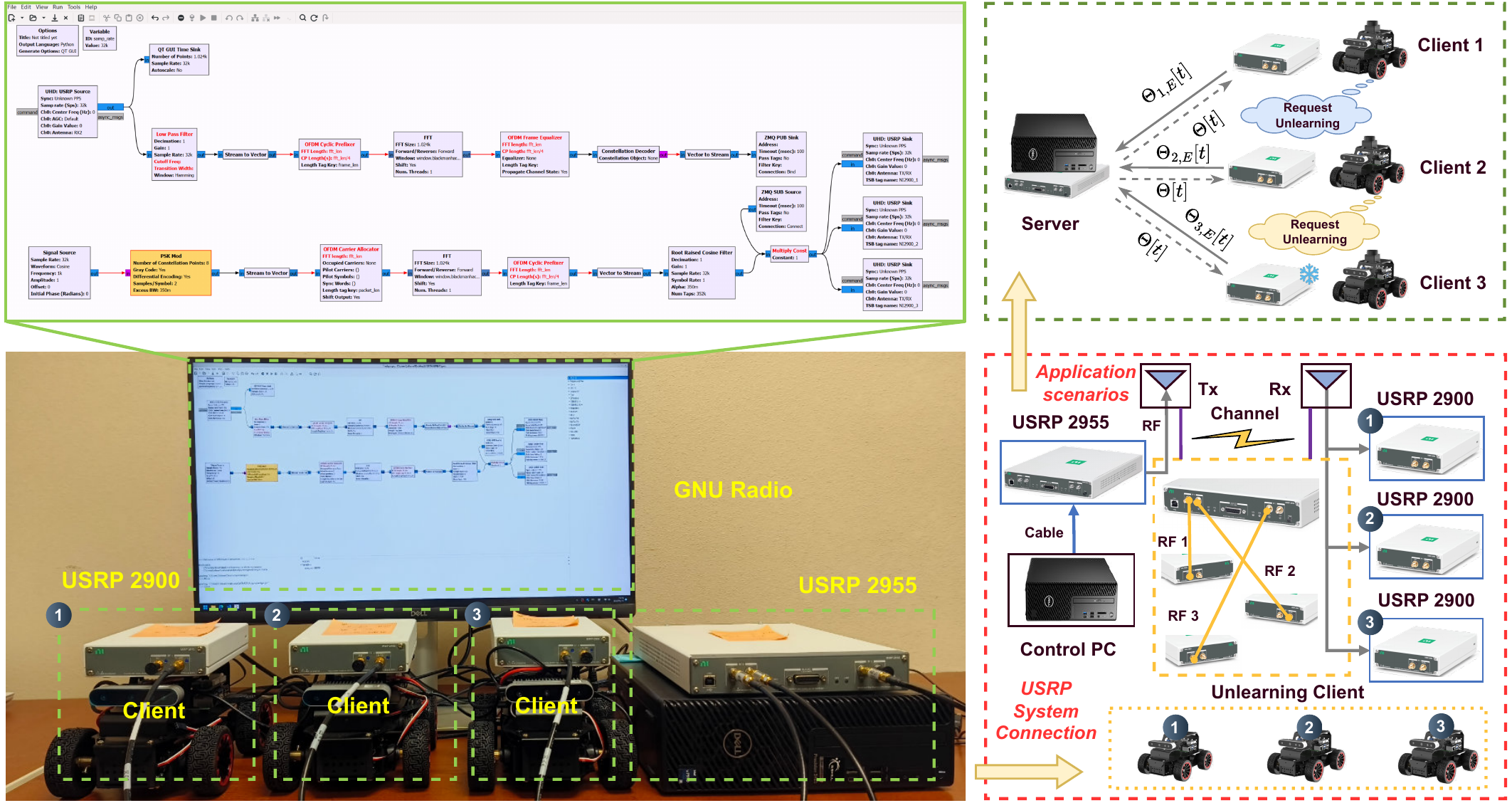}
    \caption{Testbed setup for real-time MEC FU validation.}
    \label{fig:mec_testbed}
\end{figure*}

\subsection{Experimental Setup}

The experimental setup follows a systematic approach to validate both the effectiveness of \texttt{SCALE} and the impact of information freshness on forgetting performance. Table~\ref{tab:key_parameters} summarizes the key parameters used throughout our experiments. More specifically, we evaluate \texttt{SCALE} under varying federation scales with 100 client , and test different levels of data heterogeneity using Dirichlet parameter $\alpha =0.3$. The Proximal Policy Optimization (PPO) algorithm is configured with standard hyperparameters, while the reward function balances forgetting effectiveness and information freshness with weights $w_f$ and $w_c$, respectively.
We chose FashionMNIST as our primary dataset due to its complexity and suitability for federated (un)learning evaluation. We additionally adopt CIFAR-100 (100 fine-grained classes) to test generalizability on a harder task. The choice of three distinct model architectures (MobileNetV3, ResNet18, and LeNet) allows us to evaluate the adaptability of \texttt{SCALE} with different computational complexities.



To validate \texttt{SCALE} in practical MEC systems, we implement a testbed using software-defined radios as the communication backhaul for model distribution and synchronization, as shown in Fig.~\ref{fig:mec_testbed}. The testbed consists of one USRP 2955 device serving as the communication gateway for the edge server and three USRP 2900 devices deployed as communication interfaces for mobile clients. The USRP 2955 is configured with a transmit power of 20 dBm and operates with a 100 MHz bandwidth, while each USRP 2900 is equipped with dual RF channels, supporting a maximum data rate of 10 Mbps. Each USRP 2900 is mounted on a MentorPi robotic vehicle, which integrates an ARM Cortex-A72 quad-core processor running at 1.5 GHz with 4GB RAM for local model training. A Control PC orchestrates the entire system through GNU Radio, which manages the wireless communication protocols and coordinates the federated learning workflow. The communication operates in the 2.4 GHz frequency band, using OFDM and QPSK modulation to transmit model parameters between clients and the server.


During testing, three MentorPi clients first collaborate with the edge server to perform federated learning multiple communication rounds, with the USRP devices handling the transmission of model weights. When an unlearning request is triggered for a target client, the edge server executes the \texttt{SCALE} algorithm, performing sensitive-layer identification and AoI-driven adaptive sparsification to remove the target data of the client from the global model. After the unlearning process converges at the server side, the resulting sparsified model is pruned and manually deployed to the MentorPi clients for edge inference. We evaluate the unlearning performance directly on these resource-constrained edge devices by measuring the remaining accuracy and forgetting accuracy to validate the practical effectiveness of \texttt{SCALE} under hardware constraints.

\subsection{Performance Metrics}
\begin{figure*}[t]
    \centering
    \begin{subfigure}[b]{0.30\textwidth}
        \centering
        \includegraphics[width=\textwidth]{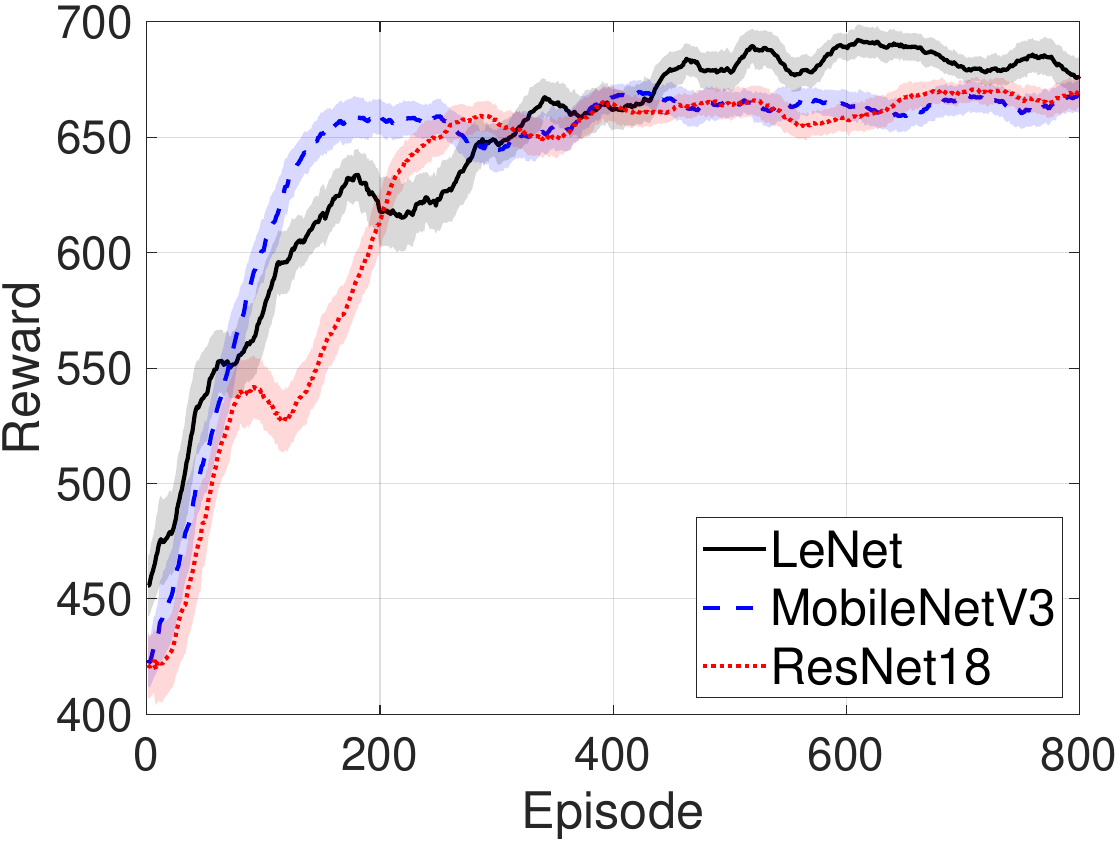}
        \caption{Client Unlearning Reward.}
        \label{fig:client_reward}
    \end{subfigure}
    \hfill
    \begin{subfigure}[b]{0.30\textwidth}
        \centering
        \includegraphics[width=\textwidth]{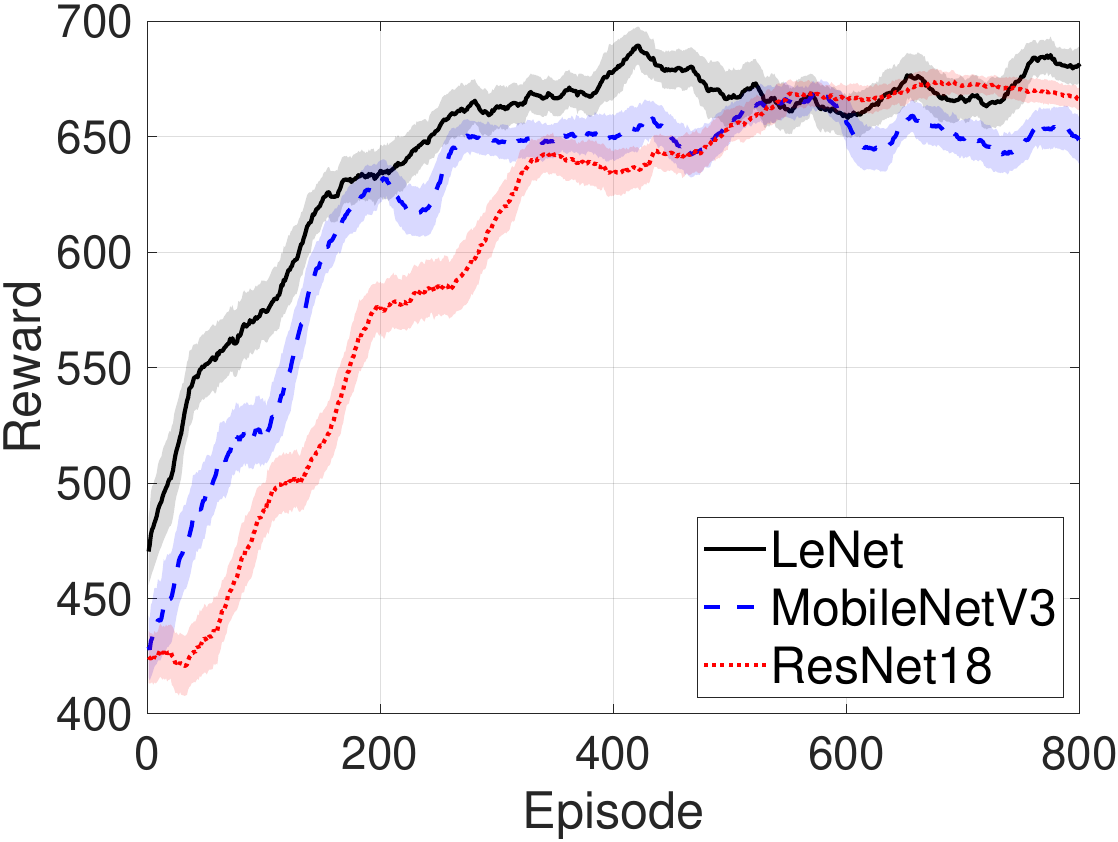}
        \caption{Class Unlearning Reward.}
        \label{fig:class_reward}
    \end{subfigure}
    \hfill
    \begin{subfigure}[b]{0.30\textwidth}
        \centering
        \includegraphics[width=\textwidth]{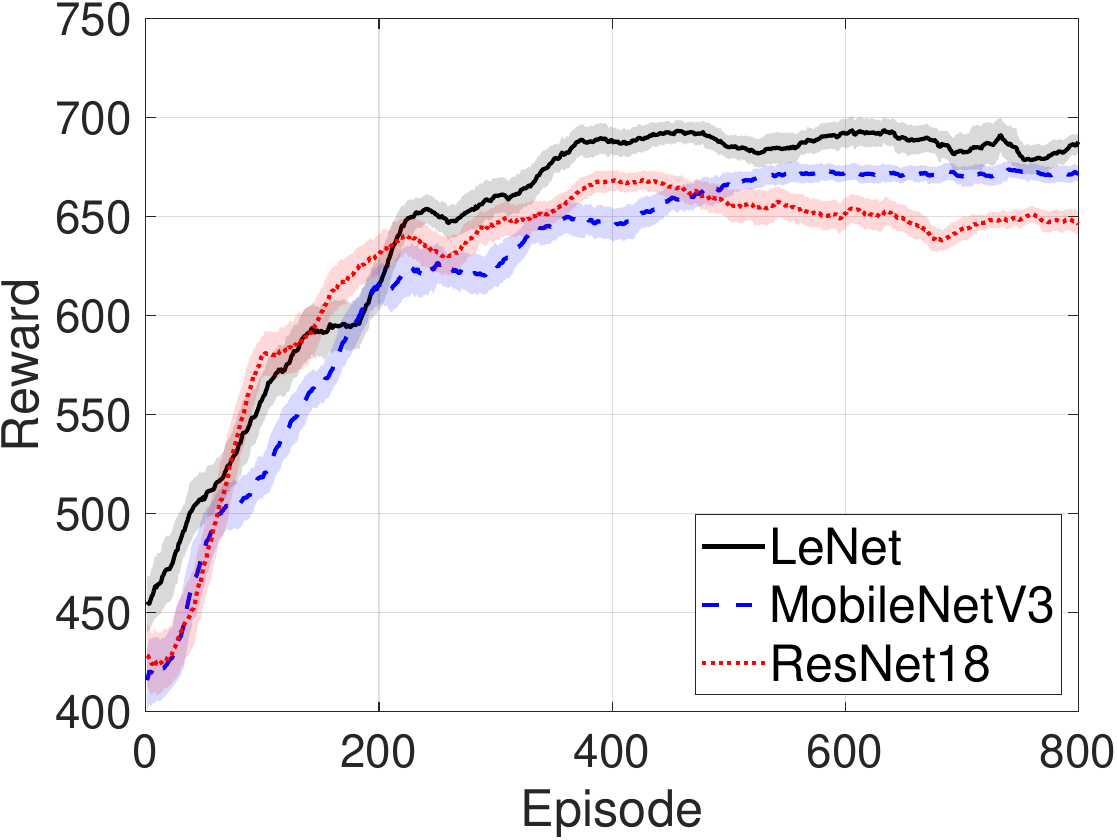}
        \caption{Sample Unlearning Reward.}
        \label{fig:sample_reward}
    \end{subfigure}
    \caption{PPO training convergence for different unlearning scenarios.}
    \label{fig:ppo_convergence}
\end{figure*}

\begin{figure*}[t]
    \centering
    \begin{subfigure}[b]{0.30\textwidth}
        \centering
        \includegraphics[width=\textwidth]{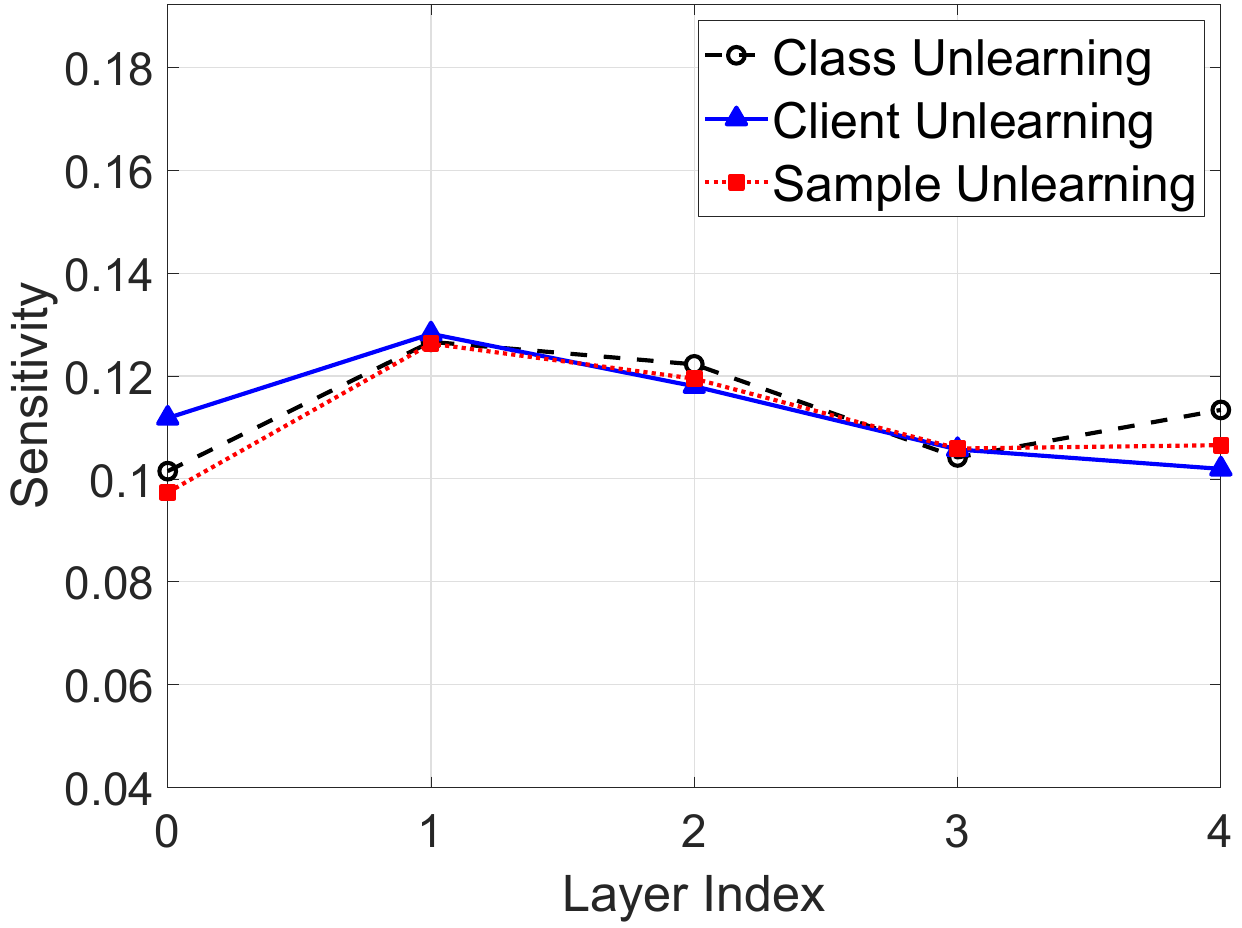}
        \caption{LeNet layer sensitivity.}
        \label{fig:lenet_sensitivity}
    \end{subfigure}
    \hfill
    \begin{subfigure}[b]{0.30\textwidth}
        \centering
        \includegraphics[width=\textwidth]{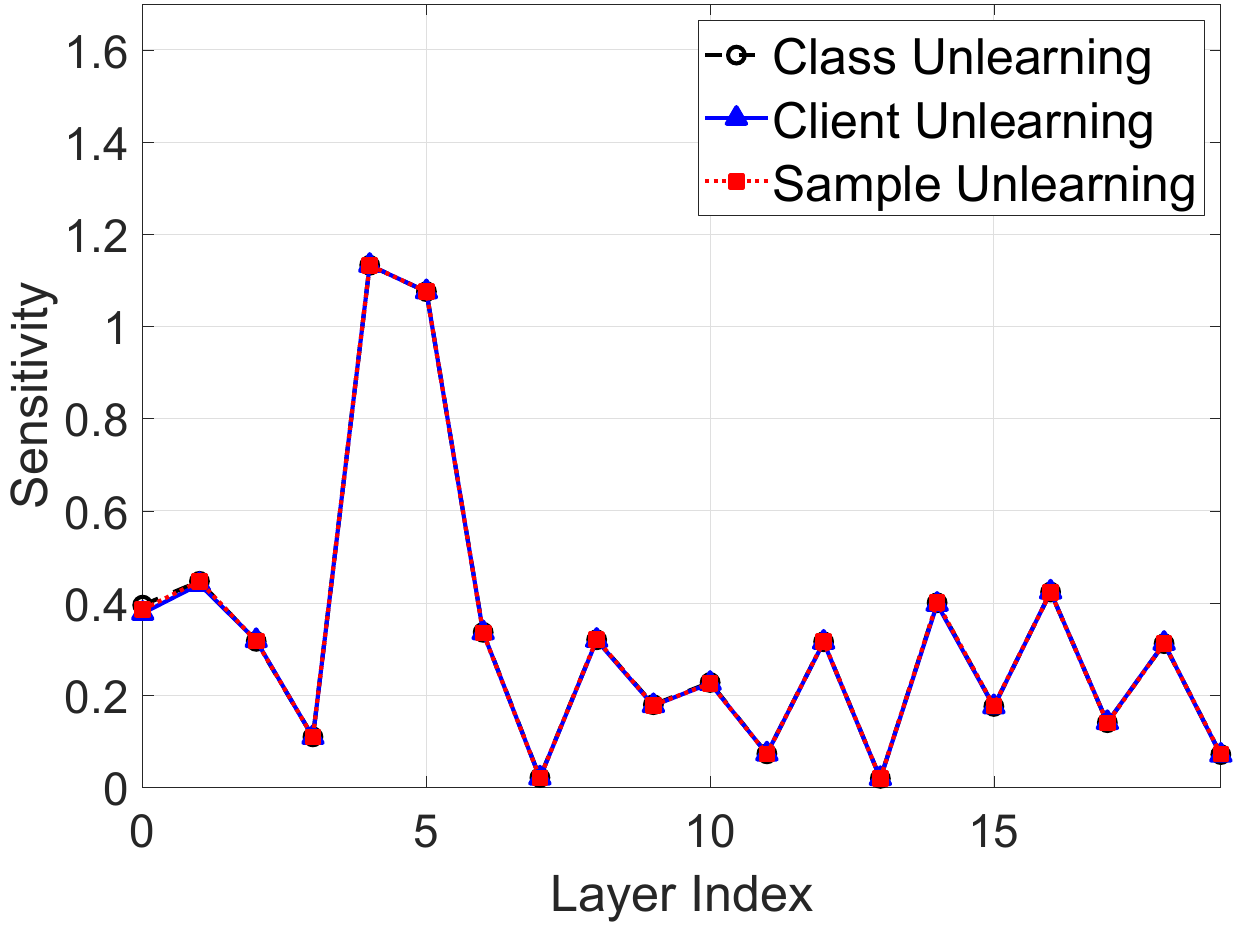}
        \caption{MobileNetV3 layer sensitivity.}
        \label{fig:mobilenet_sensitivity}
    \end{subfigure}
    \hfill
    \begin{subfigure}[b]{0.30\textwidth}
        \centering
        \includegraphics[width=\textwidth]{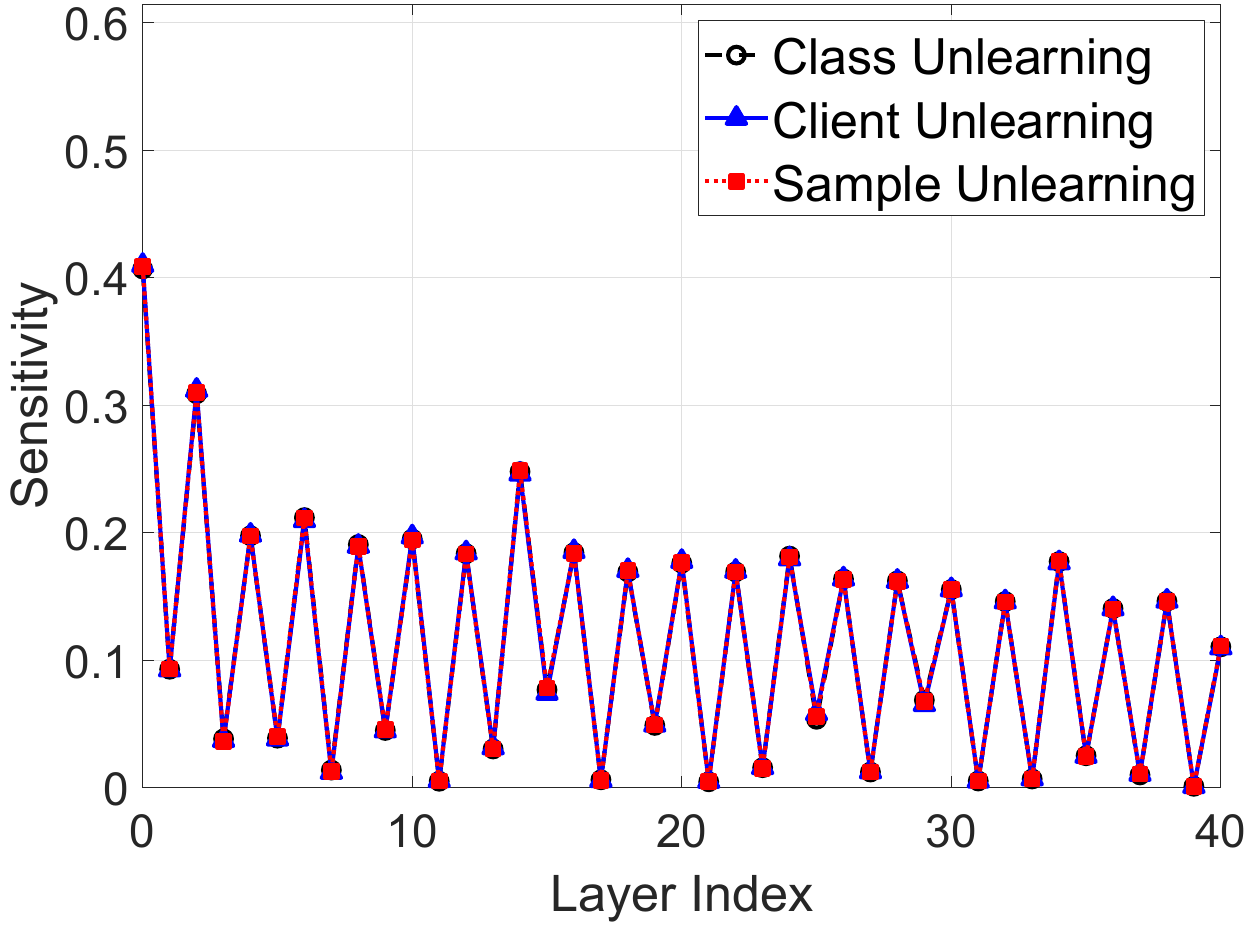}
        \caption{ResNet18 layer sensitivity.}
        \label{fig:resnet_sensitivity}
    \end{subfigure}
    \caption{Layer sensitivity analysis over different model architectures and unlearning scenarios.}
    \label{fig:layer_sensitivity}
\end{figure*}

We evaluate the federated unlearning performance in terms of the following metrics.


\textit{Remaining Accuracy (RA)}: This metric reflects whether the unlearned model maintains high accuracy on the remaining data:
\begin{equation}
RA = \frac{1}{M_r} \sum_{(x,y) \in \mathcal{D}_r} \mathbb{I}[f_{\Theta'}(x) = y],
\label{eq:remaining_accuracy}
\end{equation}

\textit{Forgetting Accuracy (FA)}: This metric indicates whether the unlearned model demonstrates poor performance on erased data:
\begin{equation}
FA = \frac{1}{M_u} \sum_{(x,y) \in \mathcal{D}_u} \mathbb{I}[f_{\Theta'}(x) = y],
\label{eq:forgetting_accuracy}
\end{equation}
where $f_{\Theta'}(x) = \arg\max_{c \in \mathcal{Y}} P(y=c|x; \Theta')$ is the prediction function, $\mathbb{I}[\cdot]$ is the indicator function that returns 1 if the condition is true and 0 otherwise, $M_r = |\mathcal{D}_r|$ and $M_u = |\mathcal{D}_u|$ are the sizes of remaining and unlearned datasets, respectively.

\textit{Forgetting Rate}: This metric measures the forgetting efficiency through prediction confidence degradation~\cite{Ma2022TDSC}:
\begin{equation}
FR = 1 - \frac{1}{M_u} \sum_{(x,y) \in \mathcal{D}_u  
} \frac{P_{\Theta'}(y|x)}{P_{\Theta}(y|x)},
\label{eq:forgetting_rate}
\end{equation}
where $P_{\Theta}(y|x)$ is the prediction confidence of the original model $\Theta$ for the true label $y$ given input $x$, $P_{\Theta'}(y|x)$ is the prediction confidence of the unlearned model $\Theta'$ for the same true label, and $M_u = |\mathcal{D}_u|$ is the size of the unlearned dataset. A higher $FR$ value indicates more effective forgetting as the model's confidence on the target data decreases.

\textit{Communication Overhead}: For fine-grained parameter manipulation, each layer $l$ into $G_l$ is decomposed into sub-parameter groups, where $j \in \{1,2,\ldots,G_l\}$ denotes the sub-group index. Each sub-group contains parameters $W_{l,j}$ that can be independently modified.

The Age of Information (AoI) for parameter sub-group $j$ in layer $l$ at time $t$ is defined as~\cite{Wu2025MNET}:
\begin{equation}
A_{l,j}[t] = t - T_{l,j},
\label{eq:aoi_def}
\end{equation}
where $T_{l,j}$ is the timestamp of the most recent update for sub-group $j$ in layer $l$.

The communication overhead is the combination of transmission costs with information freshness:
\begin{equation}
\min \quad \mathcal{C}(\Theta \rightarrow \Theta') = \alpha \cdot C_t + \beta \cdot \mathbb{E}[A_g[t]],
\label{eq:efficiency_obj}
\end{equation}
where $C_t$ is the total communication cost, $A_g[t] = \frac{1}{M \sum_{l=1}^{L} G_l} \sum_{l=1}^{L} \sum_{j=1}^{G_l} A_{l,j}[t]$ is the global system AoI, $M$ is the number of layers selected for unlearning, and $\alpha, \beta > 0$ are weighting coefficients.

\textit{Layer Sensitivity}: This metric examines whether the selected top-$M$ sensitive layers from Eq.~\eqref{eq:sensitive_layers} effectively capture the most client-influenced components.

\subsection{Results}

\begin{table*}[h]
  \centering
  \caption{Performance comparison of \texttt{SCALE} on FashionMNIST and CIFAR-100 across three unlearning scenarios. Metrics are reported as value ($\Delta$), where $\Delta$ is the gap with Retrain baseline. \underline{Underlined} values denote \texttt{SCALE} results and \textbf{bold} values denote baseline method results. $\uparrow$ indicates higher is better, $\downarrow$ indicates lower is better.}
  \label{tab:main_results}
  \resizebox{\textwidth}{!}{
  \setlength{\tabcolsep}{1.5pt}
  \begin{tabular}{l|c@{\hspace{3pt}}c@{\hspace{3pt}}c@{\hspace{3pt}}c@{\hspace{3pt}}c|c@{\hspace{3pt}}c@{\hspace{3pt}}c@{\hspace{3pt}}c@{\hspace{3pt}}c|c@{\hspace{3pt}}c@{\hspace{3pt}}c@{\hspace{3pt}}c@{\hspace{3pt}}c}
    \toprule
    & \multicolumn{5}{c|}{\textbf{RA} ($\uparrow$)} & \multicolumn{5}{c|}{\textbf{FA} ($\downarrow$)} & \multicolumn{5}{c}{\textbf{Avg. AoI (s)} ($\downarrow$)} \\
    \cmidrule(lr){2-6} \cmidrule(lr){7-11} \cmidrule(lr){12-16}
    \textbf{Model} & \textbf{Retrain} & \textbf{F-E} & \textbf{PGA} & \textbf{FUSED} & \textbf{SCALE} & \textbf{Retrain} & \textbf{F-E} & \textbf{PGA} & \textbf{FUSED} & \textbf{SCALE} & \textbf{Retrain} & \textbf{F-E} & \textbf{PGA} & \textbf{FUSED} & \textbf{SCALE} \\
    \midrule
    \multicolumn{16}{c}{\textit{Client Unlearning}} \\
    \midrule
    \multicolumn{16}{l}{\textit{FashionMNIST}} \\
    LeNet & 0.92 & 0.74 \textbf{(0.18)} & 0.78 \textbf{(0.14)} & 0.99 \textbf{(0.07)} & 0.82 \underline{(0.10)}
    & 0.00 & 0.12 \textbf{(0.12)} & 0.09 \textbf{(0.09)} & 0.00 \textbf{(0.00)} & 0.00 \underline{(0.00)}
    & 2.59 & 3.00 \textbf{(0.41)} & 2.66 \textbf{(0.07)} & 1.65 \textbf{(0.94)} & 2.01 \underline{(0.58)} \\

    MobileNetV3 & 0.92 & 0.81 \textbf{(0.11)} & 0.83 \textbf{(0.09)} & 0.81 \textbf{(0.11)} & 0.86 \underline{(0.06)}
    & 0.01 & 0.11 \textbf{(0.10)} & 0.12 \textbf{(0.11)} & 0.04 \textbf{(0.03)} & 0.00 \underline{(0.01)}
    & 2.66 & 2.91 \textbf{(0.25)} & 2.71 \textbf{(0.05)} & 1.74 \textbf{(0.92)} & 2.16 \underline{(0.50)} \\

    ResNet18 & 0.96 & 0.81 \textbf{(0.15)} & 0.84 \textbf{(0.12)} & 0.83 \textbf{(0.13)} & 0.87 \underline{(0.09)}
    & 0.04 & 0.21 \textbf{(0.17)} & 0.18 \textbf{(0.14)} & 0.08 \textbf{(0.04)} & 0.01 \underline{(0.03)}
    & 6.75 & 7.18 \textbf{(0.43)} & 6.82 \textbf{(0.07)} & 4.39 \textbf{(2.36)} & 5.99 \underline{(0.76)} \\

    \cmidrule(l){1-16}
    \multicolumn{16}{l}{\textit{CIFAR-100}} \\
    LeNet & 0.38 & 0.30 \textbf{(0.08)} & 0.32 \textbf{(0.06)} & 0.42 \textbf{(0.04)} & 0.36 \underline{(0.02)}
    & 0.01 & 0.18 \textbf{(0.17)} & 0.14 \textbf{(0.13)} & 0.02 \textbf{(0.01)} & 0.01 \underline{(0.00)}
    & 2.85 & 3.21 \textbf{(0.36)} & 2.93 \textbf{(0.08)} & 2.13 \textbf{(0.72)} & 2.21 \underline{(0.64)} \\

    MobileNetV3 & 0.62 & 0.51 \textbf{(0.11)} & 0.54 \textbf{(0.08)} & 0.55 \textbf{(0.07)} & 0.58 \underline{(0.04)}
    & 0.02 & 0.16 \textbf{(0.14)} & 0.13 \textbf{(0.11)} & 0.06 \textbf{(0.04)} & 0.02 \underline{(0.00)}
    & 2.93 & 3.18 \textbf{(0.25)} & 2.99 \textbf{(0.06)} & 1.89 \textbf{(1.04)} & 2.34 \underline{(0.59)} \\

    ResNet18 & 0.71 & 0.58 \textbf{(0.13)} & 0.60 \textbf{(0.11)} & 0.61 \textbf{(0.10)} & 0.65 \underline{(0.06)}
    & 0.05 & 0.24 \textbf{(0.19)} & 0.21 \textbf{(0.16)} & 0.10 \textbf{(0.05)} & 0.04 \underline{(0.01)}
    & 7.28 & 7.71 \textbf{(0.43)} & 7.34 \textbf{(0.06)} & 4.71 \textbf{(2.57)} & 6.41 \underline{(0.87)} \\

    \midrule
    \multicolumn{16}{c}{\textit{Class Unlearning}} \\
    \midrule
    \multicolumn{16}{l}{\textit{FashionMNIST}} \\
    LeNet & 0.90 & 0.73 \textbf{(0.17)} & 0.76 \textbf{(0.14)} & 0.99 \textbf{(0.09)} & 0.81 \underline{(0.09)}
    & 0.08 & 0.27 \textbf{(0.19)} & 0.20 \textbf{(0.12)} & 0.04 \textbf{(0.04)} & 0.00 \underline{(0.08)}
    & 2.72 & 3.07 \textbf{(0.35)} & 2.64 \textbf{(0.08)} & 1.68 \textbf{(1.04)} & 2.09 \underline{(0.63)} \\

    MobileNetV3 & 0.91 & 0.78 \textbf{(0.13)} & 0.81 \textbf{(0.10)} & 0.73 \textbf{(0.18)} & 0.87 \underline{(0.04)}
    & 0.04 & 0.17 \textbf{(0.13)} & 0.15 \textbf{(0.11)} & 0.21 \textbf{(0.17)} & 0.02 \underline{(0.02)}
    & 2.79 & 3.14 \textbf{(0.35)} & 2.88 \textbf{(0.09)} & 1.73 \textbf{(1.06)} & 2.19 \underline{(0.60)} \\

    ResNet18 & 0.96 & 0.78 \textbf{(0.18)} & 0.80 \textbf{(0.16)} & 0.73 \textbf{(0.23)} & 0.86 \underline{(0.10)}
    & 0.00 & 0.19 \textbf{(0.19)} & 0.16 \textbf{(0.16)} & 0.08 \textbf{(0.08)} & 0.00 \underline{(0.00)}
    & 6.89 & 7.08 \textbf{(0.19)} & 6.78 \textbf{(0.11)} & 4.52 \textbf{(2.37)} & 6.02 \underline{(0.87)} \\

    \cmidrule(l){1-16}
    \multicolumn{16}{l}{\textit{CIFAR-100}} \\
    LeNet & 0.36 & 0.28 \textbf{(0.08)} & 0.30 \textbf{(0.06)} & 0.40 \textbf{(0.04)} & 0.34 \underline{(0.02)}
    & 0.10 & 0.32 \textbf{(0.22)} & 0.24 \textbf{(0.14)} & 0.05 \textbf{(0.05)} & 0.01 \underline{(0.09)}
    & 2.97 & 3.31 \textbf{(0.34)} & 2.85 \textbf{(0.12)} & 1.60 \textbf{(1.37)} & 2.29 \underline{(0.68)} \\

    MobileNetV3 & 0.61 & 0.49 \textbf{(0.12)} & 0.52 \textbf{(0.09)} & 0.46 \textbf{(0.15)} & 0.57 \underline{(0.04)}
    & 0.05 & 0.21 \textbf{(0.16)} & 0.19 \textbf{(0.14)} & 0.25 \textbf{(0.20)} & 0.03 \underline{(0.02)}
    & 3.04 & 3.39 \textbf{(0.35)} & 3.13 \textbf{(0.09)} & 1.87 \textbf{(1.17)} & 2.37 \underline{(0.67)} \\

    ResNet18 & 0.69 & 0.55 \textbf{(0.14)} & 0.57 \textbf{(0.12)} & 0.49 \textbf{(0.20)} & 0.62 \underline{(0.07)}
    & 0.01 & 0.23 \textbf{(0.22)} & 0.20 \textbf{(0.19)} & 0.11 \textbf{(0.10)} & 0.01 \underline{(0.00)}
    & 7.42 & 7.62 \textbf{(0.20)} & 7.31 \textbf{(0.11)} & 4.84 \textbf{(2.58)} & 6.47 \underline{(0.95)} \\

    \midrule
    \multicolumn{16}{c}{\textit{Sample Unlearning}} \\
    \midrule
    \multicolumn{16}{l}{\textit{FashionMNIST}} \\
    LeNet & 0.89 & 0.73 \textbf{(0.16)} & 0.74 \textbf{(0.15)} & 0.99 \textbf{(0.10)} & 0.77 \underline{(0.12)}
    & 0.02 & 0.15 \textbf{(0.13)} & 0.13 \textbf{(0.11)} & 0.05 \textbf{(0.03)} & 0.09 \underline{(0.07)}
    & 2.69 & 3.03 \textbf{(0.34)} & 2.75 \textbf{(0.06)} & 1.65 \textbf{(1.04)} & 2.06 \underline{(0.63)} \\

    MobileNetV3 & 0.88 & 0.76 \textbf{(0.12)} & 0.79 \textbf{(0.09)} & 0.75 \textbf{(0.13)} & 0.82 \underline{(0.06)}
    & 0.01 & 0.18 \textbf{(0.17)} & 0.14 \textbf{(0.13)} & 0.18 \textbf{(0.17)} & 0.10 \underline{(0.09)}
    & 2.81 & 3.29 \textbf{(0.48)} & 2.92 \textbf{(0.11)} & 1.71 \textbf{(1.10)} & 2.68 \underline{(0.13)} \\

    ResNet18 & 0.92 & 0.80 \textbf{(0.12)} & 0.81 \textbf{(0.11)} & 0.54 \textbf{(0.38)} & 0.84 \underline{(0.08)}
    & 0.01 & 0.22 \textbf{(0.21)} & 0.17 \textbf{(0.16)} & 0.13 \textbf{(0.12)} & 0.10 \underline{(0.09)}
    & 6.92 & 7.12 \textbf{(0.20)} & 6.85 \textbf{(0.07)} & 4.52 \textbf{(2.40)} & 5.99 \underline{(0.93)} \\

    \cmidrule(l){1-16}
    \multicolumn{16}{l}{\textit{CIFAR-100}} \\
    LeNet & 0.34 & 0.27 \textbf{(0.07)} & 0.28 \textbf{(0.06)} & 0.39 \textbf{(0.05)} & 0.31 \underline{(0.03)}
    & 0.04 & 0.19 \textbf{(0.15)} & 0.16 \textbf{(0.12)} & 0.07 \textbf{(0.03)} & 0.12 \underline{(0.08)}
    & 2.89 & 3.27 \textbf{(0.38)} & 2.96 \textbf{(0.07)} & 1.53 \textbf{(1.36)} & 2.24 \underline{(0.65)} \\

    MobileNetV3 & 0.58 & 0.47 \textbf{(0.11)} & 0.50 \textbf{(0.08)} & 0.45 \textbf{(0.13)} & 0.53 \underline{(0.05)}
    & 0.03 & 0.22 \textbf{(0.19)} & 0.18 \textbf{(0.15)} & 0.21 \textbf{(0.18)} & 0.13 \underline{(0.10)}
    & 3.07 & 3.54 \textbf{(0.47)} & 3.18 \textbf{(0.11)} & 2.21 \textbf{(0.86)} & 2.86 \underline{(0.21)} \\

    ResNet18 & 0.66 & 0.54 \textbf{(0.12)} & 0.55 \textbf{(0.11)} & 0.32 \textbf{(0.34)} & 0.58 \underline{(0.08)}
    & 0.02 & 0.26 \textbf{(0.24)} & 0.21 \textbf{(0.19)} & 0.16 \textbf{(0.14)} & 0.13 \underline{(0.11)}
    & 7.46 & 7.66 \textbf{(0.20)} & 7.39 \textbf{(0.07)} & 4.15 \textbf{(3.31)} & 6.43 \underline{(1.03)} \\

    \bottomrule
  \end{tabular}
  }
\end{table*}

\begin{figure*}[ht]
    \centering
    \begin{subfigure}[b]{0.30\textwidth}
        \centering
        \includegraphics[width=\textwidth]{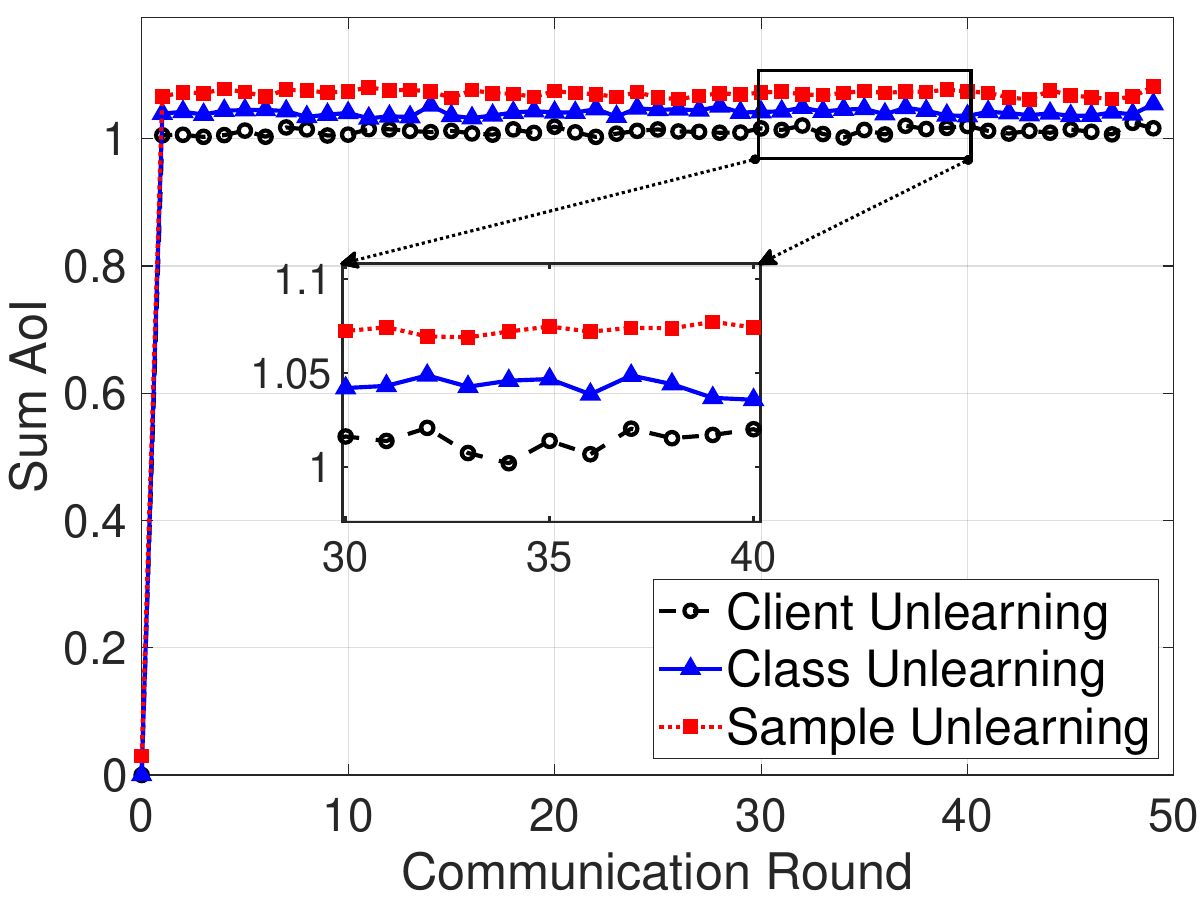}
        \caption{Sum AoI with LeNet.}
        \label{fig:sum_aoi_lenet}
    \end{subfigure}
    \hfill
    \begin{subfigure}[b]{0.30\textwidth}
        \centering
        \includegraphics[width=\textwidth]{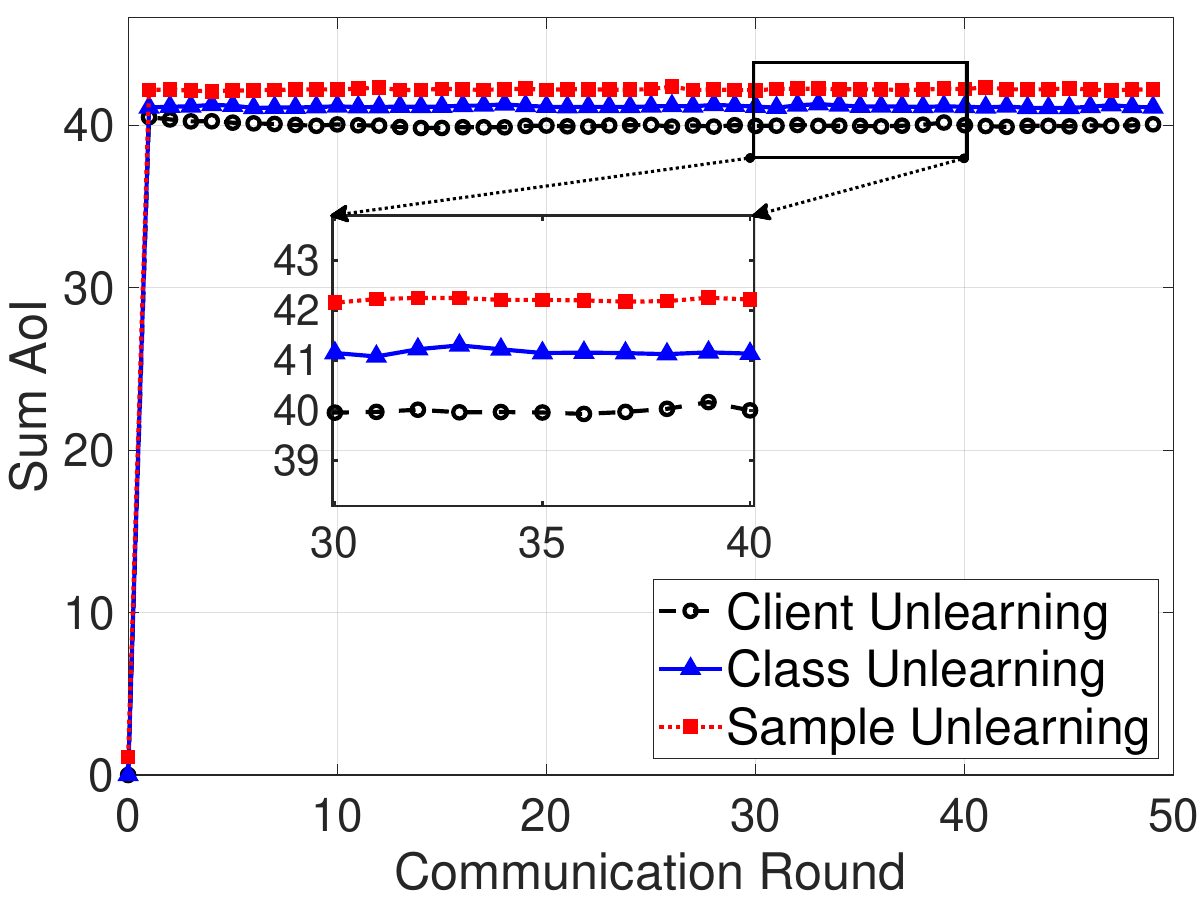}
        \caption{Sum AoI with MobileNetV3.}
        \label{fig:sum_aoi_mobilenetv3}
    \end{subfigure}
    \hfill
    \begin{subfigure}[b]{0.30\textwidth}
        \centering
        \includegraphics[width=\textwidth]{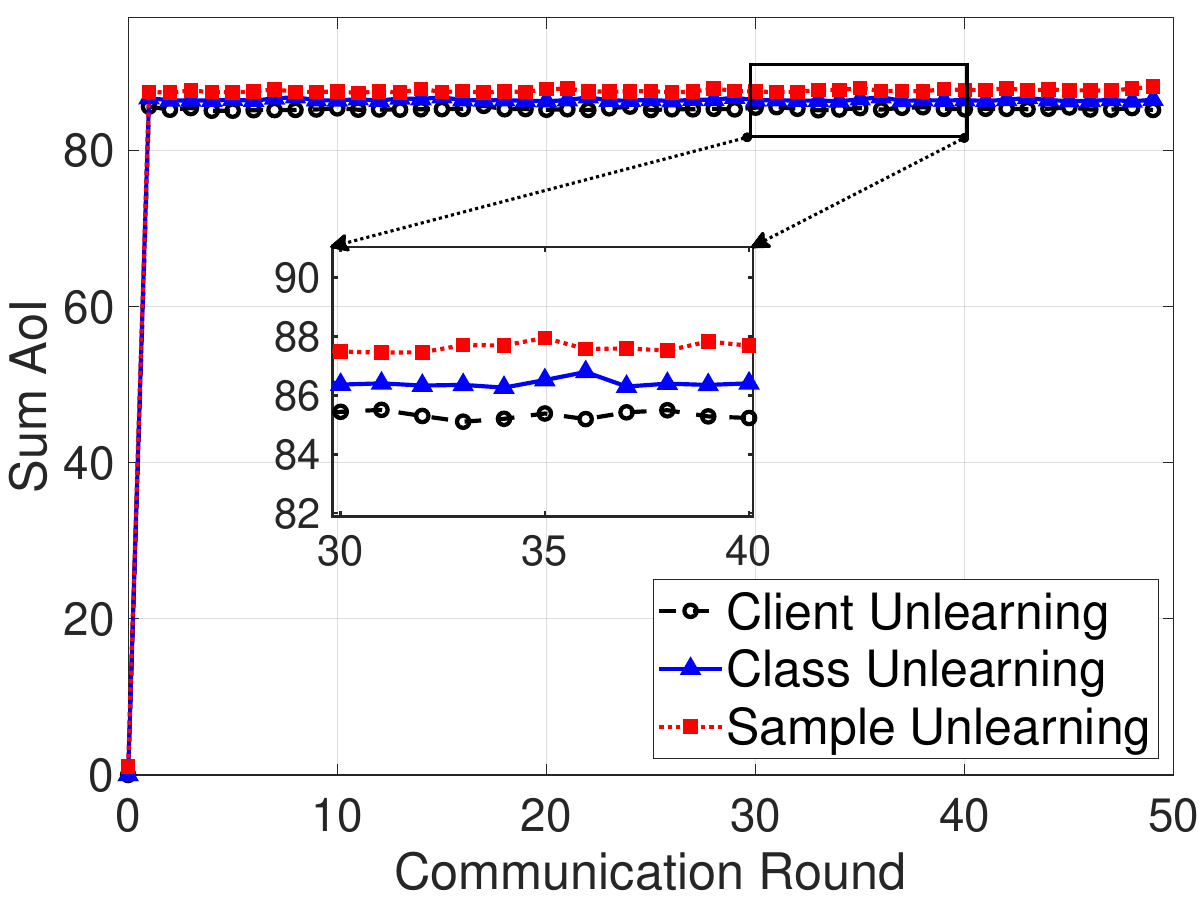}
        \caption{Sum AoI with ResNet18.}    
        \label{fig:sum_aoi_resnet18}
    \end{subfigure}
    \caption{Sum AoI comparison for different unlearning scenarios across model architectures.}
    \label{fig:sum_aoi_comparison}
\end{figure*}  

Fig.~\ref{fig:ppo_convergence} shows the PPO training over three unlearning scenarios. All models achieve stable convergence with consistent reward improvements. Notably, client unlearning exhibits the most stable convergence with minimal oscillations, while sample unlearning shows higher volatility due to its fine-grained nature. LeNet demonstrates the fastest initial convergence, whereas ResNet18 exhibits a distinctive two-phase pattern with rapid initial gains followed by fine-tuning stabilization. 

Fig.~\ref{fig:layer_sensitivity} presents a layer-wise sensitivity analysis across three model architectures (LeNet, MobileNetV3, ResNet18) under three unlearning scenarios (client unlearning, class unlearning, and sample unlearning). Remember that higher sensitivity values indicate that the corresponding layer contains parameters that are more significantly affected by the client-server data exchange and target information, thus requiring prioritized sparsification to achieve effective unlearning while minimizing impact on the remaining model performance. In particular, Fig.~\ref{fig:lenet_sensitivity} reveals that LeNet exhibits the highest sensitivity in the initial convolutional layers; Fig.~\ref{fig:mobilenet_sensitivity} demonstrates that MobileNetV3 shows localized sensitivity peaks within the intermediate depthwise separable convolution layers; and Fig.~\ref{fig:resnet_sensitivity} illustrates that ResNet18 displays a more distributed sensitivity pattern across multiple residual layers.


Table~\ref{tab:main_results} presents comprehensive comparisons between \texttt{SCALE} and the state-of-the-art methods, including Retrain~\cite{Liu2022INFOCOM}, FedEraser (F-E)~\cite{Liu2021WQOS}, Projected Gradient Ascent (PGA)~\cite{Anisa2023arxiv}, and FUSED~\cite{Zhong2025CVPR} across three unlearning scenarios on both FashionMNIST and CIFAR-100. As observed, \texttt{SCALE} demonstrates superior performance compared to F-E and PGA while maintaining good efficiency compared to FUSED.


For client unlearning on FashionMNIST, \texttt{SCALE} achieves $RA$ values of 0.82, 0.86, and 0.87 across LeNet, MobileNetV3, and ResNet18, substantially outperforming F-E (0.74, 0.81, 0.81) and PGA (0.78, 0.83, 0.84). The Average AoI for \texttt{SCALE} (2.01s, 2.16s, 5.99s) is in the between Retrain (2.59s, 2.66s, 6.75s) and FUSED (1.65s, 1.74s, 4.39s), while F-E exhibits the highest overhead at 3.00s, 2.91s, and 7.18s. Similar observations can be made in class and sample unlearning scenarios, where \texttt{SCALE} consistently maintains better $RA$ and competitive AoI compared to F-E and PGA. Although FUSED occasionally demonstrates better $RA$ and lower AoI, \texttt{SCALE} achieves a more balanced trade-off between unlearning effectiveness and model utility across diverse architectures, which makes it more suitable for real-world deployment. The same pattern appears on CIFAR-100: \texttt{SCALE} stays within $0.02$--$0.08$ of Retrain in $RA$ across all scenarios, while FUSED collapses on ResNet18 sample unlearning, where its $RA$ drops to $0.32$ while Retrain holds at $0.66$.

\subsection{Testbed Results and Discussion}

\begin{figure*}[h]
	\centering
	\includegraphics[width=0.8\textwidth]{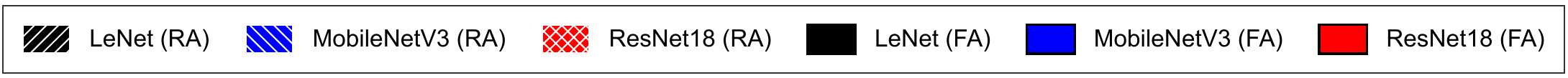}
	\vspace{5pt}
	\begin{subfigure}[b]{0.30\textwidth}
		\centering
		\includegraphics[width=\textwidth]{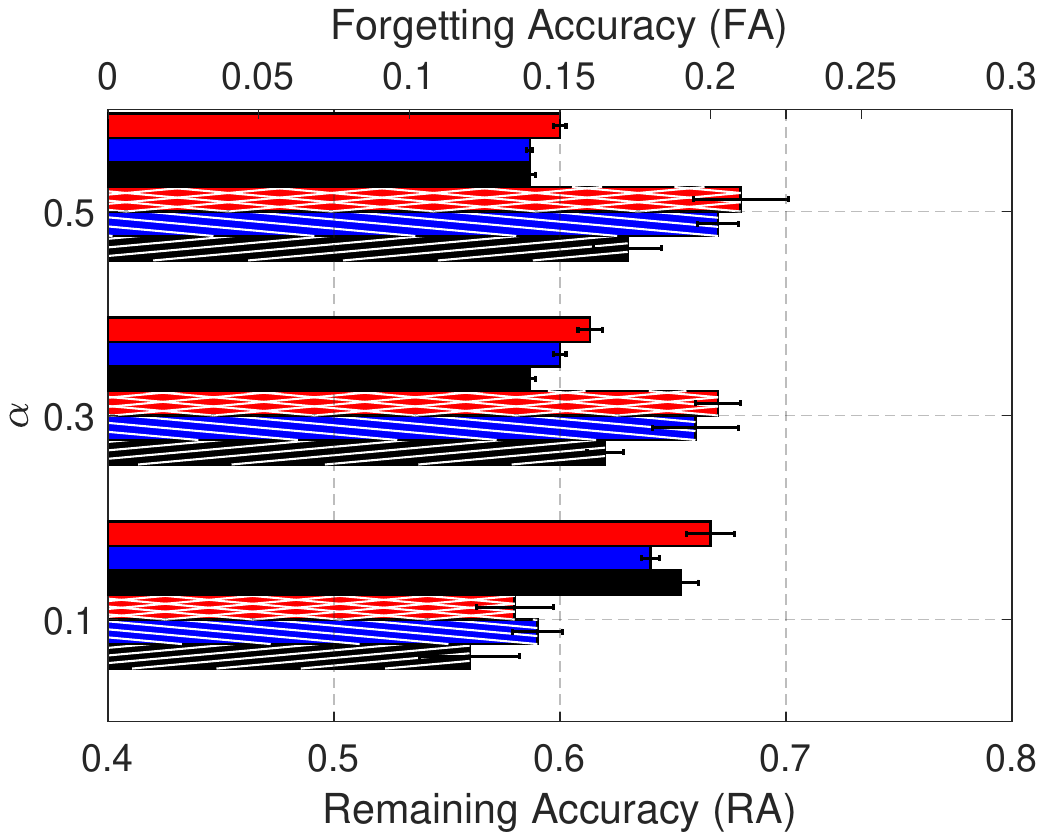}
		\caption{Client unlearning.}
		\label{fig:testbed_client_alpha}
	\end{subfigure}
	\hfill
	\begin{subfigure}[b]{0.30\textwidth}
		\centering
		\includegraphics[width=\textwidth]{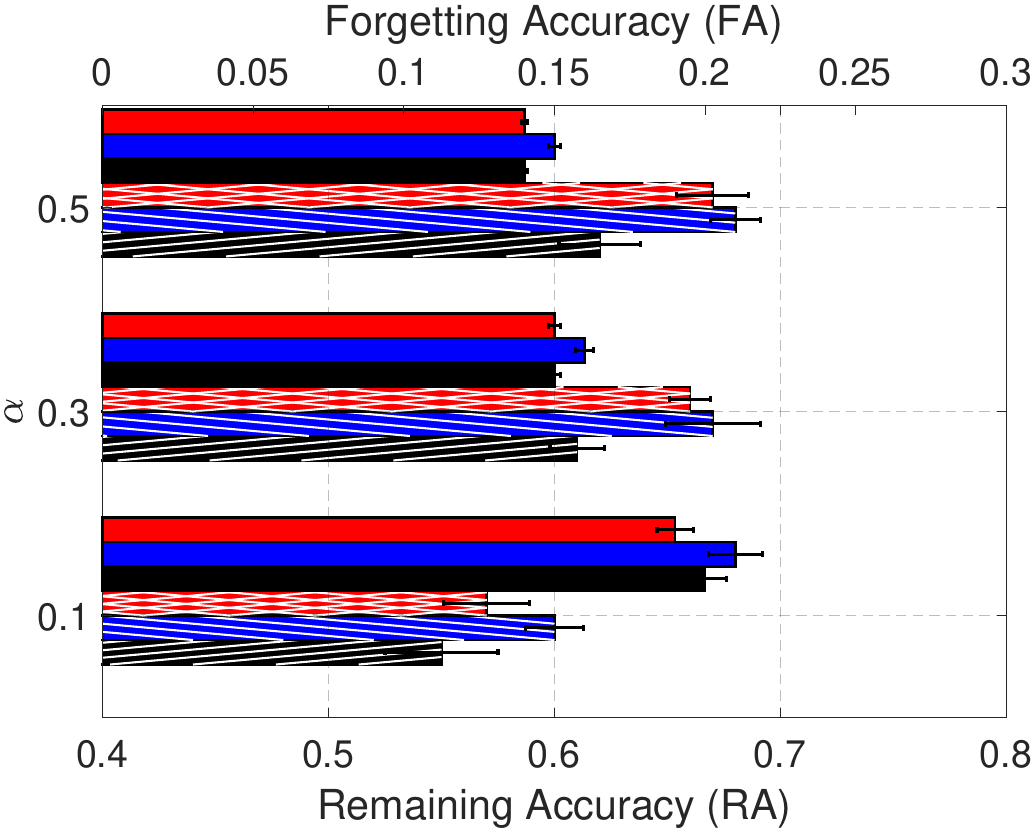}
		\caption{Class unlearning.}
		\label{fig:testbed_class_alpha}
	\end{subfigure}  
	\hfill
	\begin{subfigure}[b]{0.30\textwidth}
		\centering
		\includegraphics[width=\textwidth]{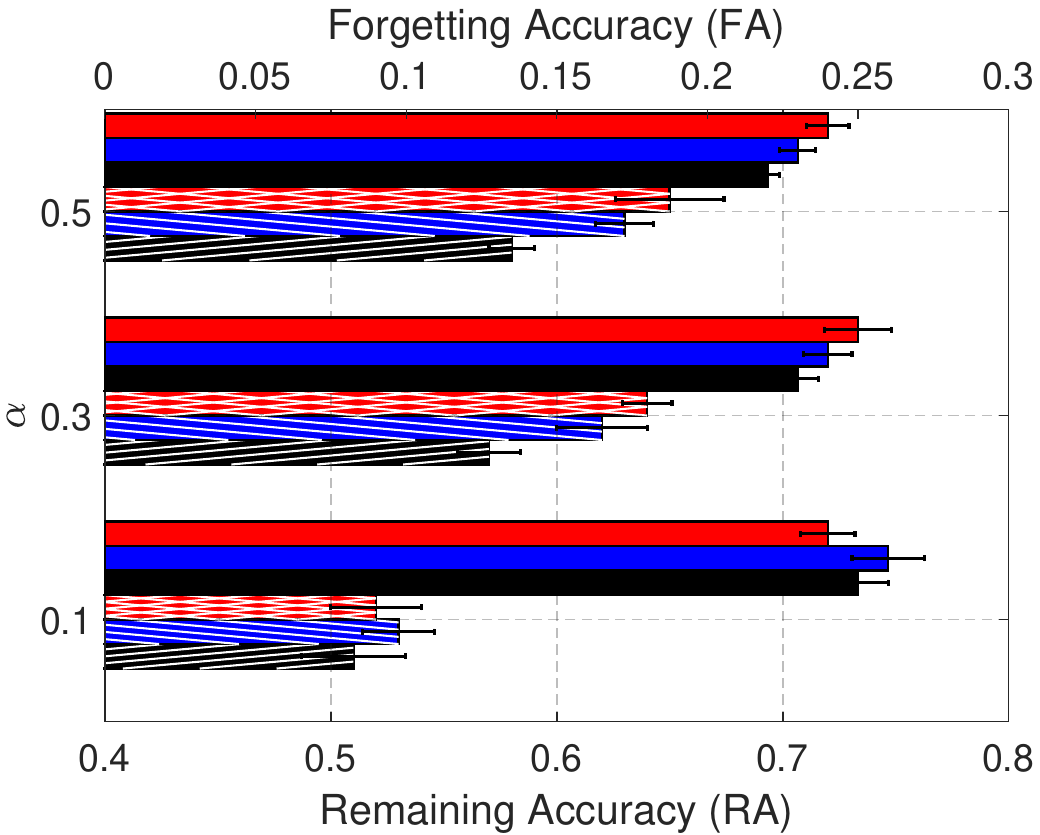}
		\caption{Sample unlearning.}
		\label{fig:testbed_sample_alpha}
	\end{subfigure}
	\caption{Impact of data heterogeneity on testbed unlearning performance in three scenarios.}
	\label{fig:testbed_heterogeneity}
\end{figure*}

Fig.~\ref{fig:sum_aoi_comparison} illustrates the temporal evolution of Sum AoI to the overall system AoI for each communication round throughout the unlearning process on three architectures. The results demonstrate consistent optimization dynamics under \texttt{SCALE} implementation, with LeNet achieving the fastest AoI stabilization, followed by MobileNetV3 and ResNet18. From the AoI performance perspective, LeNet demonstrates the lowest communication overhead with 2.013s average AoI, making it the most suitable for deployment in time-sensitive federated environments. Notably, the AoI fluctuations in physical experiments exhibit reduced variance compared to simulation results. This can be attributed to the limited client configuration, which minimizes network congestion and synchronization overhead. 


Fig.~\ref{fig:testbed_heterogeneity} presents the impact of data heterogeneity on unlearning performance across three scenarios in our hardware testbed. To validate \texttt{SCALE} under realistic hardware constraints, we deploy a subset of FashionMNIST images onto MentorPi robotic vehicles and evaluate unlearning performance using pruned model architectures optimized for resource-constrained edge devices. The Dirichlet parameter $\alpha$ controls the degree of non-IID data distribution, where smaller $\alpha$ values indicate higher data heterogeneity among clients.

As shown in Fig.~\ref{fig:testbed_client_alpha}, client unlearning performance shows strong dependence on data heterogeneity. Under highly heterogeneous conditions ($\alpha=0.1$), LeNet, MobileNetV3, and ResNet18 achieve RA values of 0.56, 0.59, and 0.58, respectively, while FA of 0.19, 0.18, and 0.20 indicate substantial forgetting difficulty with degraded model utility. As $\alpha$ increases to 0.3, RA improves to 0.62-0.67 as FA decreases to 0.14-0.16 across models. At $\alpha=0.5$, a more balanced data distribution, all models achieve RA of 0.63-0.68 with FA reduced to 0.14-0.15, which demonstrates substantially improved unlearning effectiveness. ResNet18 shows the most significant improvement, with RA increasing from 0.58 to 0.68 as $\alpha$ increases from 0.1 to 0.5. These observations reveal that balanced data distributions facilitate more effective client removal while preserving model utility, as reduced overlap in client-specific feature representations enables cleaner parameter separation.


Fig.~\ref{fig:testbed_class_alpha} shows similar patterns for class unlearning. 
The consistent improvement across $\alpha$ values validates that uniform class distribution across clients supports more targeted class removal with minimal impact on remaining knowledge.

Note that the sample unlearning results in Fig.~\ref{fig:testbed_sample_alpha} exhibit similar heterogeneity sensitivity but with overall degraded performance compared to coarser-grained scenarios. At $\alpha=0.1$, LeNet, MobileNetV3, and ResNet18 achieve RA of 0.51, 0.53, and 0.52 with FA of 0.25, 0.26, and 0.24, respectively. Even at $\alpha=0.5$, RA only reaches 0.58-0.65 while FA remains at 0.22-0.24, substantially higher than the 0.14-0.15 FA observed in client and class unlearning. ResNet18 demonstrates the largest RA improvement (from 0.52 to 0.65), but its FA shows minimal reduction. This performance gap stems from the fine-grained nature of sample removal, in which individual data influences are diffusely distributed across parameter space, making precise localization and elimination significantly more challenging regardless of data distribution characteristics.


Results across all three scenarios consistently show that data heterogeneity is a key factor influencing the effectiveness of federated unlearning in MEC environments. Our testbed validation indicates that \texttt{SCALE} achieves stable performance under different levels of heterogeneity on resource-constrained edge devices. In particular, balanced data distributions ($\alpha=0.5$) provide the best trade-off between forgetting effectiveness ($FA=0.14-0.24$) and model utility preservation.

\section{Conclusion} \label{sec:conclusion}

In this paper, we have presented \texttt{SCALE}: a novel dual-level federated unlearning framework that addresses the issues of low unlearning precision and lack of temporal information freshness in mobile edge computing systems. By integrating historical contribution-based layer sensitivity analysis with Age-of-Information-driven adaptive sparsification, \texttt{SCALE} effectively balances unlearning effectiveness with information freshness while maintaining computational efficiency. Our theoretical analysis demonstrates the convergence properties and acceleration advantages of the proposed framework. At the same time, experiments based on multiple neural network architectures validate the outstanding performance of \texttt{SCALE} compared to existing baselines. The testbed implementation further confirms the practical feasibility of our framework in real-world MEC environments. We believe that \texttt{SCALE} offers a powerful solution for federated unlearning in MEC that meets stringent regulatory requirements while preserving model utility and system performance.

\end{document}